\documentclass[pra,twocolumn,notitlepage,superscriptaddress]{revtex4-1}
\usepackage{epsfig,color}
\usepackage{graphicx}
\usepackage{dcolumn}
\usepackage{bm}
\usepackage{amsmath, amsfonts, amssymb,mathrsfs}
\usepackage{pstricks}
\usepackage{amsxtra}
\usepackage{amsthm}
\usepackage{natbib}
\usepackage{qcircuit}
\usepackage{array}

\newcommand{\ket}[1]{\mbox{$|#1\rangle$}}
\newcommand{\bra}[1]{\mbox{$\langle#1|$}}

\def\be{\begin{equation}}      
\def\ee{\end{equation}}
\def\beu{\begin{equation*}}   
\def\eeu{\end{equation*}}

\providecommand{\abs}[1]{\left\lvert#1\right\rvert}   
\DeclareMathOperator{\trace}{Tr}      

\providecommand{\mean}[1]{\langle#1\rangle}

\providecommand{\del}{\partial}
\providecommand{\br}{{\bm{r}}}
\providecommand{\bx}{{\bm{x}}}
\providecommand{\by}{{\bm{y}}}

\newcommand{\E}{\mathcal{E}}
\providecommand{\id}{{\mathbb{I}}}
\newcommand{\M}{{\mathcal{M}}}
\newcommand{\K}{{\mathcal{K}}}
\newcommand{\N}{{\mathbb{N}}}
\newcommand{\C}{{\mathbb{C}}}

\newtheorem{theorem}{Theorem}[section]
\theoremstyle{definition}

\definecolor{new}{rgb}{.08,.05,.8}

\newcommand{\delete}[1]{{}}

\begin{document}
\title{Entanglement Structure of Current-Driven Diffusive Fermion Systems}
\date{\today}
\author{Michael J. Gullans}
\author{David A. Huse}
\affiliation{Department of Physics, Princeton University, Princeton, New Jersey 08544, USA}
\begin{abstract}
When an extended system is coupled at its opposite boundaries to two reservoirs at different temperatures or chemical potentials, it cannot achieve a global thermal equilibrium and is instead driven to a set of current-carrying nonequilibrium states.  
Despite the broad relevance of such a scenario to metallic systems, there have been limited investigations of the entanglement structure of the resulting long-time states, in part, due to the fundamental difficulty in solving realistic models for disordered, interacting electrons.  We investigate this problem by carefully analyzing two ``toy'' models for coherent quantum transport of diffusive fermions: the celebrated three-dimensional, noninteracting Anderson model and a class of random quantum circuits acting on a chain of qubits, which exactly maps to a diffusive, interacting fermion problem. Crucially, the random circuit model can also be tuned to have no interactions between the fermions, similar to the Anderson model.  We show that the long-time states of driven noninteracting fermions exhibit volume-law mutual information and entanglement, both for our random circuit model and for the nonequilibrium steady-state of the Anderson model.  With interactions, the random circuit model is quantum chaotic and approaches local equilibrium, with only short-range entanglement. These results provide a generic picture for the emergence of local equilibrium in current-driven quantum-chaotic systems, and also provide examples of stable, highly-entangled many-body states out of equilibrium.  We discuss experimental techniques to probe these effects in low-temperature mesoscopic wires or ultracold atomic gases.
\end{abstract}
\maketitle

\section{Introduction}

Uncovering general principles that describe the entanglement structure of quantum many-body systems is a fundamental challenge in statistical mechanics and quantum information science \cite{Horodecki09}. 
In the ground state of local Hamiltonian systems, the entanglement entropy often satisfies an ``area law,'' whereby the entropy of a subregion scales with the area of its boundary \cite{Bombelli86,Srednicki93,Wolf08,Eisert10}.  Single eigenstates with finite energy density above the ground state typically exhibit extensive entanglement entropy \cite{Deutsch91,Srednicki94,DAlessio16};  however, the mutual information of finite-temperature thermal Gibbs (thus, {\it mixed}) states  still exhibits an area law \cite{Wolf08}.  
The existence of such area laws allows a rich set of analytical and  numerical tensor network techniques to be used to characterize and classify these states \cite{White92,Schollwock05,Verstraete04}.  
Although area laws for the mutual information are typical of thermal mixed states, any modification that drives the system out of equilibrium allows for potential violations.   

 A common nonequilibrium scenario consists of an extended system  coupled to two reservoirs with different chemical potentials, which drives currents in the system.    The analog of thermalization in these systems is the approach to local equilibrium at long times.
Motivated by recent developments in the understanding of thermalization and many-body localization in closed quantum many-body systems \cite{DAlessio16,Nandkishore15}, we revisit this class of current-driven open quantum systems with the goal of determining the entanglement structure of the long-time density matrix of the system.  Our results hold the most physical significance for disordered metallic systems where one can identify four length scales governing the qualitative features of this  nonequilibrium transport problem: the elastic mean free path $\ell$, the phase-coherence length $\ell_{\varphi}$, the electron-electron energy relaxation length $\ell_{\rm ee}$, and the electron-phonon scattering length $\ell_{\rm ep}$.  In a typical metal at low temperatures $\ell < \ell_{\varphi} < \ell_{\rm ee} < \ell_{\rm ep}$ \cite{Altshuler82}.  

Rather than attempting to analyze realistic models for disordered, interacting electrons, we instead focus our investigations on two ``toy'' models: the three-dimensional, noninteracting Anderson model in the diffusive phase \cite{Anderson58}, and a class of random quantum circuits acting on a chain of qubits, which exactly map to a many-fermion system.  The current-driven Anderson model is expected to qualitatively capture the properties of mesoscopic disordered metals in the regime  $L \ll \ell_{\rm \varphi}$, where $L$ is the length of the driven system.

In this article, we show that the long-time density matrix for driven noninteracting fermions is characterized by volume-law mutual information and entanglement, in distinct contrast to the entanglement properties at equilibrium.   This result should  generally apply to driven systems in the regime  $\ell \ll L \ll \ell_{\rm \varphi}$.  Our results thus provide examples of physical systems where volume-law entanglement can be sustained, and possibly harnessed, despite strong coupling of the system to external reservoirs.   We discuss experimental methods to probe these effects in transport experiments on mesoscopic wires or ultracold Fermi gases.    In our random circuit models we can also add interactions and show the crossover to area-law entanglement as $L$ exceeds $\ell_{\varphi}$, showing that such quantum chaotic-driven systems stay closer to local equilibrium than do the noninteracting fermion models.

Many-body physics models based on random quantum circuits  have recently attracted interest in high-energy \cite{Hayden07,Sekino08,Lashkari11} and quantum condensed matter physics \cite{Nahum16,Nahum17,vonKeyserlingk17,Khemani17,Tibor17} because they exhibit the quantum chaotic dynamics and rapid scrambling characteristic of interacting many-body systems, while still retaining a simple enough structure to allow controlled calculations of various measures of entanglement and chaos.  In the context of our work, the random circuit model we introduce is advantageous because it provides an analytically tractable realization of key qualitative features of the entanglement structure of these current-driven diffusive systems.  We study a generalization of these models where the interaction is tunable, allowing us to explore both the noninteracting fermion regime $L\ll\ell_{\varphi}$, the strongly chaotic regime $L\gg\ell_{\varphi}$, and the crossover between these regimes.

An additional, seemingly unrelated, aspect of the random circuit model is that it has a strong connection to an exactly solvable classical, boundary-driven stochastic lattice gas model for diffusion called the symmetric-simple-exclusion process (SSEP) \cite{KipnisLandim99,Derrida07}.   Thus, our results also provide insights into the emergence of classical hydrodynamics from interacting quantum many-body systems.  

 The paper is organized as follows:  In Sec.~\ref{sec:summary} we give a more detailed overview of the main results.  In Sec.~\ref{sec:circuit}  we present a detailed analysis of the random circuit model, including its operator spreading dynamics,  phase diagram, and entanglement properties.   In Sec.~\ref{sec:ext} we discuss some natural extensions of the random circuit model to higher dimensions and to more than two states per site.
 In Sec.~\ref{sec:anderson} we analyze the mutual information in the current-driven Anderson model using a scattering-state approach and compare these results to a scattering-state analysis of the random circuit.   In Sec.~\ref{sec:appLE} we show that signatures of the volume-law entangled phase of the random circuit also show up in intermediate-time dynamics following a quantum quench.  In Sec.~\ref{sec:exp} we discuss a method to experimentally probe signatures of the volume-law entanglement in transport through  mesoscopic wires or  ultracold Fermi gases.
 We present our conclusions in Sec.~\ref{sec:outlook}.  In the appendixes, we present several useful technical results.  
 See Appendix \ref{app:overview} for an overview.

\section{Summary of main results}
 \label{sec:summary}
 
 As described in the Introduction, the two main results in this paper are (i) the discovery that noninteracting, current-driven, diffusive fermion models exhibit extensive mutual information and entanglement and (ii) the development of a simple physical picture for how interactions effectively decohere these correlations and recover the expected area-law scaling of local equilibrium.    The essential  arguments underlying these results are based on universal properties of operator-spreading dynamics and, thus, should apply to a wide range of physical systems and models.  To provide deeper insight into our results, we systematically analyze the behavior and phenomenology of three different classes of models that exhibit these universal features: random quantum circuits, noninteracting Anderson models, and weakly interacting, disordered metals.   In this section, we summarize our findings for each of these classes of models.
 
   \begin{figure}[tb]
\begin{center}
\includegraphics[width = 0.49 \textwidth]{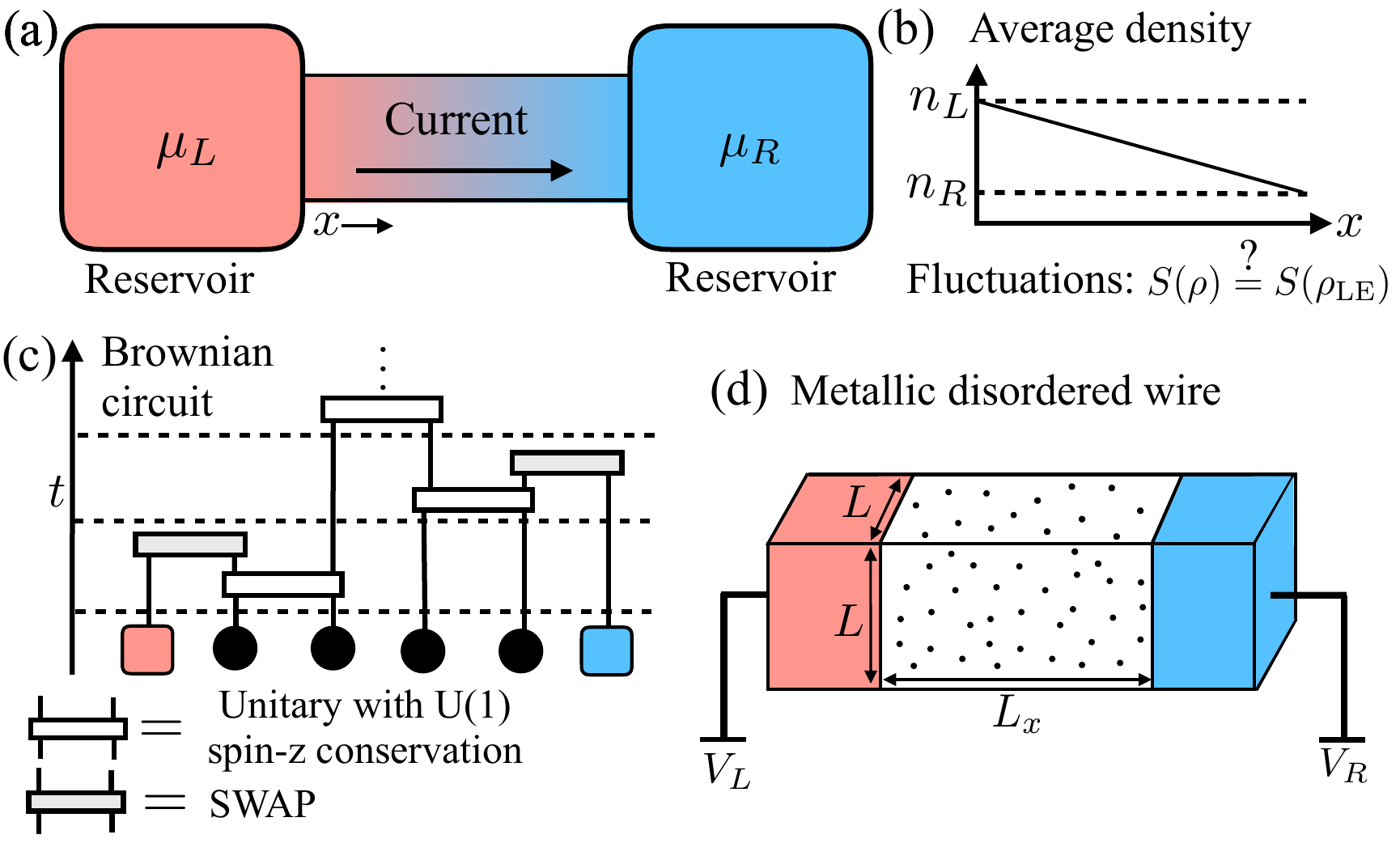}
\caption{(a) Current-driven model that we consider, where two reservoirs are held at fixed chemical potentials and the bulk dynamics conserves the total density.  (b) The average local density will generically be linear in position, but the total von Neumann entropy may have an extensive (``volume-law'') deviation from local equilibrium (LE).  (c) Schematic of the random circuit model, which exhibits quantum chaotic behavior for some parameter regimes.  The interaction with reservoir is produced by a SWAP gate that acts immediately after any other gate is applied to the boundary spin and swaps the boundary spin state with a fresh spin state from the reservoir. (d) We also consider the diffusive regime of noninteracting fermions in a random potential in 3D and find that they display a similar volume-law entanglement structure as the noninteracting fermion regime of the random circuit model.}
\label{fig:model}
\end{center}
\end{figure}
 
The general setup we consider is shown in Fig.~\ref{fig:model}(a) and consists of a system with a global conservation law (either magnetization or particle number) in contact with two thermodynamic reservoirs at different chemical potentials.  The chemical potential bias leads to steady-state currents in the long-time limit.  At long times and wavelengths, the models we consider have an effective hydrodynamic description (derived below) for the average value of the conserved quantity $n(\br,t)$ given by the diffusion equation
\be
\frac{\del }{\del t} n(\br,t) = \tilde D \nabla^2 n(\br,t),
\ee
where $\tilde D$ is the scaled diffusion constant and $\br$ is a $d$-dimensional position vector scaled so that $x=0/1$ corresponds to the longitudinal position of the left/right reservoir.  When subject to the boundary condition $n(\br)\lvert_{x=0}=n_L$ and $n(\br)\lvert_{x=1}=n_R$, this equation has the steady-state solution shown in Fig.~\ref{fig:model}(b)
\be \label{eqn:nx}
n(\br)=n_L(1-x)+n_R x.
\ee
  This average profile, however, gives no information about correlations and the entanglement structure in the long-time states.  To more systematically investigate the entropy and entanglement properties of such current-driven quantum systems, we analyze the models discussed above and shown schematically in Figs.~\ref{fig:model}(c)-(d).

 \subsection{Random circuit model}
 In the random circuit model  [see Fig.~\ref{fig:model}(c)], our system is a spin-1/2 chain of $L$ sites with the $z$ component of the magnetization of the spins being the ``charge'' that is conserved and that is being transported by the current.  The dynamics  are generated by a sequence of total spin-$z$ conserving, randomly chosen nearest-neighbor unitary operations, i.e., quantum gates, applied at each time step, with no correlations between time steps.   Whenever a gate operates on a boundary spin, that boundary spin is immediately swapped with a ``fresh'' (uncorrelated) spin from the adjacent infinite reservoir; this is how the system is coupled to the reservoirs.    The unitary operations in the local two-site basis $\lvert \uparrow\uparrow\rangle,\lvert \uparrow\downarrow\rangle,\lvert \downarrow\uparrow\rangle,\lvert \downarrow\downarrow\rangle$ take the  form \cite{Khemani17,Tibor17}
 \be \label{eqn:U}
U = \left( \begin{array}{c  c c c} U_+ & 0 & 0 &0 \\ 
0 & U_{ud} & U_{rl} & 0 \\
0 & U_{lr} & U_{du}&0 \\
0&0&0&U_{-}  \end{array} \right).
\ee
  The coefficients in $U$ are chosen according to a two-parameter family of random distributions that we describe in Sec. \ref{sec:circuit}.  One can understand the time evolution under the random circuit as a stochastic process whereby in each time step a Hamiltonian for the system is  randomly chosen from a given distribution and applied for a fixed length of time.
The average dynamics of this class of random circuit models are diffusive because  (i) we choose the sites where we apply the gates with a uniform probability and (ii) unitarity constrains $|U_{rl}| = |U_{lr}|$.  As a result, at each time step, the local spin-density always has equal probability of hopping to the left or right.  
  Due to the diffusive transport to the memoryless reservoirs, the long-time density matrix for a given circuit realization is insensitive to its initial conditions and converges to a particular time-dependent mixed state that we call a ``nonequilibrium attracting state'' (NEAS), in analogy to the nomenclature of a nonequilibrium steady-state (NESS) for driven time-independent systems.  
  
  More precisely, after $N=L t$ random gates are applied (note that we scale time $t$ so that $L$ gates of the circuit are applied in one unit of time), we can describe the action of the random circuit on the initial density matrix of the system $\rho_{I}$ by a linear operator acting on the space of density matrices 
  \be \label{eqn:Mrho}
  \begin{split}
   \mathcal{M}_{ t}(\rho_I) &= \trace_{\rm res}(U_{L t} \cdots U_1 \rho_I \otimes \rho_{\rm res} U_1^\dag \cdots U_{Lt}^\dag) \\
   &= \trace_{\rm res}(U_{Lt} \mathcal{M}_{ t-\delta t}(\rho_I) \otimes \rho_{\rm res}^{LR} U_{L t}^\dag),
   \end{split}
\ee
where $\delta t= 1/L$, $U_i$ are the randomly chosen unitary gates at each time step that include the SWAPs with the fresh reservoir spins,  $\rho_{\rm res}$ is the initial many-body density matrix of the reservoir,  $\trace_{\rm res}(\cdot)$ denotes a partial trace over the  reservoirs, and $\rho_{\rm res}^{LR} = \rho_{m_L} \otimes \rho_{m_R}$ is the two-site density matrix of the fresh reservoir spins that can become entangled with the system by $U_{i}$.  Here we have defined
\be \label{eqn:rhomlr}
 \rho_{m_{L/R}} = \left( \begin{array}{c c}
 \frac{1}{2} + m_{L/R} & 0 \\
 0 & \frac{1}{2} - m_{L/R} 
 \end{array} \right),
\ee
where $m_{L/R}$ are the magnetizations of the left/right reservoirs.
    For $t$ much larger than the diffusive transit time through the system, i.e., the Thouless time $\tau_{\rm Th} = L^2/D$, the action of $\mathcal{M}_{ t}$ takes the form 
\be
 \mathcal{M}_{ t} (\rho_I) = \rho_{\rm NEAS}(t) \trace[\rho_I] , 
\ee
  where $\rho_{\rm NEAS}(t)$ depends on the  history of the circuit, 
  but is independent of $\rho_I$ in the limit of long time.
We choose the model so that the average NEAS density matrix over random circuit realizations is the same for all values of the parameters.  Interestingly, this average density matrix can be found from the known solution for the NESS of the classical stochastic lattice gas model that goes by the name of the symmetric-simple-exclusion process \cite{Derrida07}.   The entropy and mutual information of the NESS for the SSEP have been shown to have only subextensive deviations from local equilibrium  \cite{Bahadoran07,Derrida06}.

   As there is no energy conservation for the random circuit,  local equilibrium is simply given by the product-state density matrix with the same average magnetization profile as Eq.~(\ref{eqn:nx}),      
  \begin{align} \label{eqn:rhomi}
\rho_{\rm LE}& = \bigotimes_{i=1}^{L} \rho_{m_i},~\rho_{m_i} = \Big(\frac{1}{2}+m_i\Big) u_i + \Big(\frac{1}{2}-m_i\Big) d_i,\\
m_i &= m_L \Big(1- \frac{i}{L+1}\Big) + m_R  \frac{i}{L+1},
\end{align}
where $u_i = (1+\sigma_i^{z})/2$ and $d_i = (1-\sigma_i^{z})/2$ denote projections of each site $i$ onto spin up or down.

Despite this uniform average behavior, we find three qualitatively different regimes of behavior for the instantaneous NEAS $\rho_{\rm NEAS}(t)$ when we investigate the deviation of its von Neumann entropy  $S(\rho)=-\trace[\rho \log \rho]$ from $S(\rho_{\rm LE})$,
\be
\Delta S \equiv S(\rho_{\rm LE})-\overline{S(\rho_{\rm NEAS})},
\ee
where the overbar denotes an average over realizations of the random circuit.
As we show below, the generic behavior of the random circuit model is strongly interacting and quantum chaotic.   In this quantum chaotic regime, we show that the scaling of $ \Delta S$ with the length of the system $L$ is subextensive.  As a result,  the NEAS is, as one might expect, close to local equilibrium.     For certain limiting values of the parameters of the model, however, we obtain two other fine-tuned phases whose entropy deviates extensively from local equilibrium, i.e., $\Delta S \propto L$.  One of these special phases corresponds to a classical deterministic dynamics, while the other maps to noninteracting diffusive fermions.
  
The emergence of a volume-law scaling for $\Delta S$ coincides with nonchaotic dynamics in each of these two special phases.  However, we find that their mutual information and entanglement properties  have dramatically different behavior. Here, the mutual information between two regions $A$ and $B$ is defined in terms of the von Neumann entropy as $I(A:B)\equiv S(\rho_A)+S(\rho_B)-S(\rho_{AB})$, where $\rho_C$ is the reduced density matrix on region $C$.  To study the  entanglement, we use an entanglement measure appropriate for mixed states known as the logarithmic negativity \cite{Vidal02,Plenio05}.   We remark that similar to the area law for the mutual information, recently an area law for the logarithmic negativity was proved for thermal-equilibrium Gibbs states \cite{Sherman16}.

In one of these two nonchaotic phases, the  mutual information and logarithmic negativity exactly vanish.  In this phase, the extensive deviation of the entropy from local equilibrium already appears in the instantaneous magnetization profile of the system, i.e., single-site correlation functions.  
  The properties of this phase, as well as many aspects of its crossover to the quantum chaotic phase, can be understood using a completely classical description of the spin dynamics.   In the other phase, we observe both volume-law mutual information and logarithmic negativity of the NEAS density matrix.   This phase can be more easily understood after transforming into a fermion representation of the spins via a Jordan-Wigner transformation, where it has a description in terms of noninteracting fermions whose average dynamics are diffusive.  This mapping motivates us to search for a similar volume-law mutual information in the three-dimensional Anderson model, which is a paradigmatic model for diffusive, noninteracting particles.  
 We give a more complete picture for the nonequilibrium dynamics that govern the emergence of these three phases in Sec.~\ref{sec:circuit}.

\subsection{Anderson model}
As mentioned in the Introduction, the volume-law entanglement phase of the random circuit has some key qualitative similarities with current-driven disordered Anderson models.  The Anderson model we consider is a three-dimensional (3D) tight-binding model with quenched disorder governed by the Hamiltonian
\be
H= t_0 \sum_{\mean{ij}}c_i^\dagger c_j + \sum_i V_i c_i^\dagger c_i ,
\ee
where $c_i$ is a fermion operator on site $i$, $t_0$ is the hopping, the first sum is over nearest-neighbor sites of a cubic lattice, and $V_i$ is a random potential on each site.  We  draw $V_i$ from a uniform distribution between $\pm W/2$.  In 1D and 2D, any amount of disorder localizes the eigenstates, while in 3D this model has a metal-insulator transition near $W_c \approx 16.5\, t_0$ \cite{MacKinnon81,Pichard81,Slevin14}.      One distinction between the Anderson model and the random circuit is that the Anderson model is also subject to energy conservation.
    In fact, it is known from theoretical studies using semiclassical Boltzmann equations \cite{Nagaev92,Nagaev95,Kozub95} and experimental measurements \cite{Pothier97} that, for  disordered wires of length $\ell \ll L \ll \ell_{\rm ee}$, the local energy distribution function strongly deviates from local equilibrium.  Thus the question we address in the context  of the Anderson model is how the entanglement structure for the NESS density matrix within a single disorder realization compares to the disorder-averaged density matrix.  We focus on the metallic regime $0<W< W_c$ in this work, since we are interested in diffusive systems.

We model the driven problem by taking a finite disordered region connected at its two ends to ballistic leads that are otherwise identical to the central region.  The correlations in the NESS are determined by the properties of the scattering states in a narrow energy range (up to thermal broadening) between the chemical potentials of the two leads. There are two sources of correlations that then give rise to the mutual information: First, there is the range in energy difference over which the scattering-state wave functions are correlated in the disordered region.  This energy scale is simply set by the Thouless energy, which is the inverse of the diffusive transit time through the disordered region $E_{\rm Th} = \hbar D/L_x^2$
where $D$ is the diffusion constant and $L_x$ is the longitudinal length of the disordered region (assumed to be nearly equal to the transverse width $L$) \cite{Thouless74,Anderson80}.  Crucial to the existence of the diffusive phase in 3D is that for large $L$ and $L_x$ the Thouless energy is much larger than the average level spacing of about $1/L^2L_x$ that the disordered region would have if it were isolated \cite{Anderson80}.  In the context of the NESS, this scaling implies that the eigenfunctions of the open system are significantly modified due to the coupling to the leads.  In addition to these energy correlations, there are also correlations between different transverse scattering channels. 
Accounting for both these types of contributions, we find that we recover the volume-law scaling of the mutual information  predicted from the random circuit model.  Note that this effect is absent in equilibrium because the scattering states originating  from each lead are then equally populated and interfere with each other to cancel the long-range correlations present in the eigenfunctions.  For a ballistic conductor biased at its two ends, a simple calculation shows that there is no similar buildup of  volume-law mutual information in the NESS.  We provide a more detailed analysis of the Anderson model in Sec.~\ref{sec:anderson}.

\subsection{Disordered metals}
An intriguing implication of the above results is that the volume-law entangled phase observed in the random circuit model and the Anderson model should also arise in current-driven disordered metals in the ``mesoscopic'' regime $\ell \ll L \ll \ell_{\varphi}$.   In Sec.~\ref{sec:appLE} we further explore the connection between our results and interacting, disordered electron systems by considering the quench dynamics of an initial state with a large step in the density profile.  We provide qualitative arguments that, for weak interactions, the reduced density matrix over the diffusive length scale $\sqrt{D t}$ will similarly exhibit volume-law scaling of entanglement until a crossover timescale set by the inelastic scattering rate.    In Sec.~\ref{sec:exp}, we describe an  approach based on single-particle interference experiments to directly probe these effects in mesoscopic wires or atomic Fermi gases \cite{Gross17}.

 \subsection{Relation to prior work}
 
  Here, we review some prior related work on boundary-driven classical and quantum problems.
  For classical versions of this class of current-driven, nonequilibrium problems, a large body of work has been devoted to deriving emergent hydrodynamic descriptions of hard-core stochastic lattice gases \cite{KipnisLandim99,Derrida07}.  In this case, rigorous arguments have been formulated  showing that the entropy for a large class of these models converges to local equilibrium (up to subextensive corrections) \cite{Bahadoran07,Derrida06}.  At the same time,  it has been found that, even in one dimension, these boundary-driven classical systems exhibit behavior traditionally associated with critical models at equilibrium, such as power-law correlations and spontaneous symmetry breaking \cite{Gacs86,Gacs01,Evans95,Evans98}, making them a rich avenue for investigation.
Studies of quantum versions of these current-driven problems have mainly focused on integrable or free-fermion models \cite{Blythe07,Prosen09,Prosen11a,Prosen11b,Buca14,Znidar14,Zanoci16,Carollo18}.   Adding integrability-breaking terms or interactions to these lattice models generally leads to diffusive dynamics at long times, unless the system is in a many-body localized phase \cite{Meisner03,Prosen07,Huang13,Znidaric10,Znidaric18}.  
Other  work has aimed at finding efficient tensor network descriptions of the steady states of these models based on the assumption that they satisfy an area law or have an integrable structure \cite{Temme10,Brandao15,Hosur16,Mahajan16,Haegeman17}. 
A related quench problem  considers two identical  many-body  systems at different temperatures or chemical potentials suddenly brought into contact and allowed to evolve \cite{Ruelle00,Ogata02,Bhaseen15}. For integrable models, the steady state is nonthermal \cite{Castro16,Bertini16,Ljubotina17,Collura18}.   In some cases, it has been shown that there is a logarithmic violation of the area law for the mutual information \cite{Eisler14} and entanglement \cite{Coser14,Eisler14b,Hoogeveen15,Wen15,Dubail17} for this type of quench problem.

In the case of the 3D Anderson model, to our knowledge, the presence of these extensive correlations in the NESS of the current-driven problem has not been previously discussed in the literature.  There is a large body of work studying shot-noise correlations of disordered mesoscopic systems \cite{Buttiker92,Beenakker92,Blanter00};  however, the presence of such correlations between spatially separated leads follows directly from current conservation and does not provide direct information about the mutual information or entanglement.  For  free-fermion or Luttinger-liquid leads connected by a time-varying quantum point contact, there is a coincidental relation between   the full counting statistics of the current and growth of entanglement  entropy in the leads \cite{Klich09, Hsu09}.  These studies, however, considered a spatially zero-dimensional region between the leads, finding logarithmic growth of entanglement entropy in the time direction, and did not consider steady-state properties.     
Other related work has considered the wave-function entanglement in the Anderson model (i.e., entanglement entropy of the system with a single occupied eigenstate) \cite{Jia08} and the ground-state entanglement entropy of random spin chains \cite{Refael04,Refael09}, finding a logarithmic scaling with system size.

\section{Random circuit model}
\label{sec:circuit}

   In this section, we systematically analyze the nonequilibrium quantum dynamics and phenomenology of the random circuit model.
To tune between the three phases in this model, we draw the nearest-neighbor random gates 
defined in the local two-site basis $\lvert \uparrow\uparrow\rangle,\lvert \uparrow\downarrow\rangle,\lvert \downarrow\uparrow\rangle,\lvert \downarrow\downarrow\rangle$ as
 \be \label{eqn:U}
U = \left( \begin{array}{c  c c c} U_+ & 0 & 0 &0 \\ 
0 & U_{ud} & U_{rl} & 0 \\
0 & U_{lr} & U_{du}&0 \\
0&0&0&U_{-}  \end{array} \right),
\ee 
according to a two-parameter measure $d \mu$ over three gate sets $\mu_{0,1,2}$:
\be \label{eqn:dmu}
d\mu = (1-p_1-p_2) d\mu_0 + p_1 d\mu_1 + p_2 d \mu_2,
\ee
where the probabilities $p_{1,2}$ satisfying $p_1+p_2 \le 1$  are our tuning parameters.  
With probability $p_1$, we apply a ``noninteracting fermion'' (NIF) gate.  These NIF gates are then chosen as follows: We choose $\phi_0$  with uniform probability between $[0,2\pi]$ and  fix $U_+ = U_-^* = e^{i \phi_0}$.   The central 2$\times$2 matrix is then drawn from the  Haar random ensemble on SU$(2)$.  
More explicitly, it takes the form 
\begin{align}
 \left( \begin{array}{c  c}
  U_{ud} & U_{rl} \\
  U_{lr} & U_{du} 
  \end{array} \right) = \left( \begin{array}{c c}
  e^{i \phi_1} \cos \theta/2 & e^{i \phi_2} \sin \theta/2 \\
-  e^{-i \phi_2} \sin \theta/2 & e^{-i \phi_1} \cos \theta/2
\end{array} \right)
\end{align}
 with $\phi_{1,2}$ chosen uniformly between $[0,2\pi]$ and $\theta$ chosen in the interval $[0,\pi]$ with probability density $\sin \theta$.

When acting only on nearest neighbors, such a circuit can be efficiently simulated using a fermion representation of the qubits obtained by a Jordan-Wigner transformation \cite{Terhal02}.  In this case, the spin density is mapped to the fermion density.  These NIF gates are the only ones that perform a ``partial swap,'' where all four $U_{rl},U_{lr},U_{ud},U_{du}$ are nonzero.  With probability $p_2 \leq (1-p_1)$, we  apply a random unitary  chosen as follows: With equal probability, we choose a gate from one of the two ``interaction gate'' sets that produce interactions between the fermions:
 \begin{align} \label{eqn:p2gates} 
 U_1&=e^{i \phi_1} u_1 u_2 + e^{i \phi_2} u_1 d_2+ e^{i \phi_3} d_1 u_2 + d_1 d_2,\\
 U_2&= e^{i \phi_1} u_1 u_2 + e^{i \phi_2} \ell_1 r_2+ e^{i \phi_3} r_1 \ell_2 + d_1 d_2,
 \end{align}
 where $u_i$ and $d_i$ are projectors onto up and down spins, respectively, and $r_i = \sigma_i^{+}$ and $\ell_i= \sigma_i^{-}$ are single-site raising and lowering operators. Once one of these two gate sets is chosen, we then choose the $\phi_i$ with uniform probability in the interval $[0,2\pi]$.  Note that these interaction gates do not perform partial swaps, which implies that states in the $z$ basis are mapped to a single state in the $z$ basis.
Finally, with probability $1-p_1-p_2$ we apply either an iSWAP gate ($U_+=U_-=1$, $U_{ud}=U_{du}=0$ and $U_{lr}=U_{rl}=i$ so that it is in the NIF class) or the identity operation with equal probability.  This last set of gates produces neither interactions nor partial swaps.

As described in Sec.~\ref{sec:summary}, this random circuit does not produce a time-independent steady state in the long-time limit but rather induces a distribution over NEASs.   A more rigorous proof of this result is provided by three general theorems in Appendix \ref{app:thms}. The physical argument underlying the theorems is rather straightforward:
Because of the diffusive transport in the system arising from swap and partial swap gates, after a time much greater than $L^2$ ($L+1$ gates of the random circuit occur in one unit of time), the trajectory of each spin within the system has almost certainly involved a swap with a reservoir, which has no memory of the initial state within the system.  The ensemble of circuits, or the ensemble of all times for a single circuit, produces a probability density $\mathbb{P}(\rho_{\rm NEAS})$ over the NEAS density matrices $\rho_{\rm NEAS}$ that depends on the parameters $p_1$ and $p_2$.   
    
An important feature that we have designed into this family of models is that the circuit-averaged density matrix $\bar{\rho} \equiv  \int d\rho_{\rm NEAS} \mathbb{P}({\rho}_{\rm NEAS})\rho_{\rm NEAS}$ is independent of $p_1$ and $p_2$.  To see this independence, one can first derive an equation of motion for the average density matrix of the system at time $t$
\be \label{eqn:drhodt}
\bar{\rho}_t= \int d\nu \E_{\nu}(\bar{\rho}_{t-\delta t}),
\ee
where $\delta t = 1/L$ is the minimal time step in this discrete model, $\E_{\nu}(\rho) = \trace_{\rm res}(U_\nu  \rho \otimes \rho_{\rm res}^{LR}U_\nu^\dag)$ is a linear operator on the space of density matrices (often referred to as a quantum channel; see Appendix \ref{app:thms}) for a single time step of the random circuit, and $d \nu$ is a measure over all allowed gates in the random circuit that accounts for the randomness in the choice of sites to apply the gate, as well as the randomness of the two-site unitary.  Recall that $\rho_{\rm res}^{LR} = \rho_{m_L} \otimes \rho_{m_R}$ is the density matrix of the two fresh reservoir spins that can become entangled with the system by $U_\nu$ [see (\ref{eqn:rhomlr})]. 
 The quantum channel for the whole circuit, defined in Eq.~(\ref{eqn:Mrho}), is given by the composition of a long sequence of independently chosen  $\E_{\nu}$.
   Since $\E_\nu$ is a linear operator and the  late-time probability distribution for the density matrix is independent of time, we can average both sides of Eq.~(\ref{eqn:drhodt}) over $\mathbb{P}(\rho_{\rm NEAS})$ to obtain the steady-state equations for the average density matrix
\be \label{eqn:rhobar1}
\bar{\rho} = \int d\nu \E_\nu(\bar{\rho}).
\ee
Because of the random phases, the transport, and the swaps with the reservoirs, it is easy to show, using Eq.~(\ref{eqn:rhobar1}), that all off-diagonal terms in the density matrix average to zero such that
\be
\bar{\rho} = \sum_{\{ \tau_i\}} P(\tau_1,\ldots,\tau_L) \bigotimes_{i=1}^{L} [\tau_i u_i +(1- \tau_i)d_i ],
\ee
where $\tau_i \in \{0,1\}$ is a pseudospin variable for site $i$.  Moreover,  the probability measure $P({\bm \tau})$ satisfies the same steady state equation as the SSEP \cite{Derrida07}
\be \label{eqn:dpdt}
\frac{d P({\bm \tau})}{dt } = \sum_{\{\sigma_i \}} W_{\bm \tau}^{{\bm \sigma}} P({\bm \sigma}) =0~.
\ee
We give the full expression for the transition matrix $W_{\bm \tau}^{{\bm \sigma}}$ and review some basic properties of the SSEP in Appendix \ref{app:SSEP}. 
 This model is exactly solvable using a translationally invariant matrix-product-state (MPS) representation for $P(\bm{\tau})$, with a bond dimension equal to $L$ \cite{Derrida93}.  The diverging bond dimension in this solution is needed to account for the long-range correlations induced by the currents with a translationally invariant MPS.   The average spin current between two sites satisfies Fick's law and is given by $J_i = \overline{\mean{u_i - u_{i+1}}} = \delta/(L+1)$, where $\delta = m_L-m_R$, $\mean{\cdot} \equiv \trace[\rho\, \cdot]$, and the overbar denotes averages over $\mathbb{P}(\rho_{\rm NEAS})$ or, equivalently, time averages.

\begin{figure}[tb]
\begin{center}
\includegraphics[width = 0.49 \textwidth]{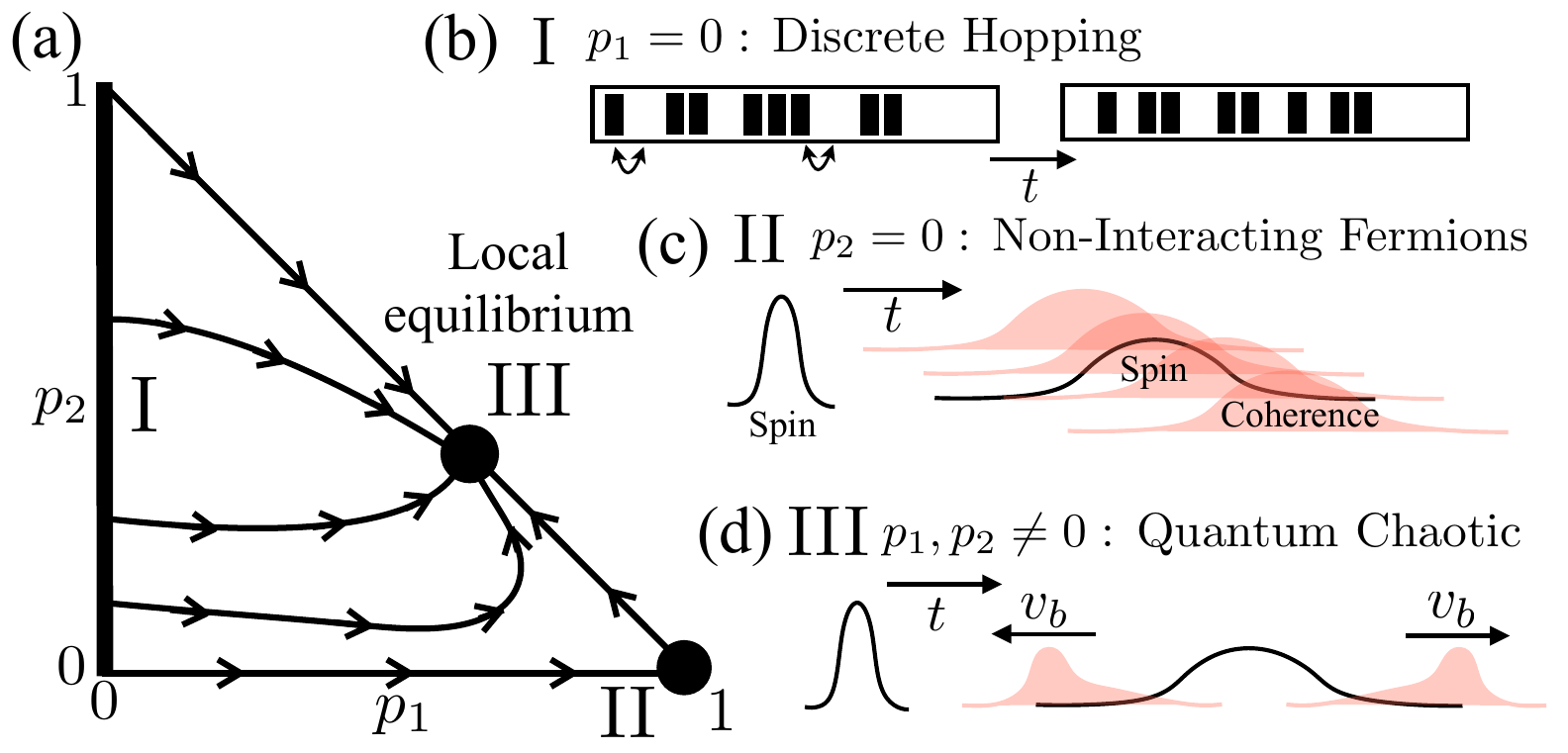}
\caption{(a) Phase diagram for small $\delta$ as a function of probabilities $p_1$ and $p_2$  determining the distribution of random unitaries in the random circuit.  We find three distinct phases in the thermodynamic limit.  Phase I is realized along the entire $p_1=0$ axis.  Phase II is realized for $p_2=0$ and $p_1 > 0$.  Turning on $p_1>0$ always destabilizes phase I, as indicated by the schematic ``flow lines,'' obtained by scaling up the system size for fixed values of the parameters (see Sec.~\ref{sec:crossovers} for a precise analysis of the crossovers).  Similarly, turning on $p_2>0$ in phase II causes a ``flow'' to phase III.
(b)   Phase I has a mapping to a classical hard-core lattice gas with discrete hopping.  (c) In phase II both the spin and the coherences move and spread diffusively. (d) In phase III, the coherences spread ballistically at the butterfly speed $v_B$, rapidly reaching the boundaries where they are ``decohered'' by the reservoirs (see Sec.~\ref{sec:lrc}).}
\label{fig:phase}
\end{center}
\end{figure}

\begin{table*}[tb]
\caption{Properties defining the three phases.  
Here, $I(L:R)$ for $\bar{\rho}$ is bounded by $\log L$.  In phases II and III, the results are derived only for small $\delta$.
}
\begin{center}
\begin{ruledtabular}
\begin{tabular}{cccccc} 
            & Transport & Operator spreading & Entropy production &$ \Delta S$         &  $\overline{I(L:R)}$  \\
\hline
I  & Diffusive/Fick's law&  Diffusive &$ \sqrt{t}$  & Volume  & 0  \\
II &          $''$          & Diffusive &  $\sqrt{t}$ & Volume & Volume \\

III &      $''$            & Ballistic &$ t \to \sqrt{t} $&  Area &  Area \\ 
\end{tabular}
\end{ruledtabular}
\end{center}
\label{tab:1}
\end{table*}%

The central results for the random circuit model can be summarized via the phase diagram shown in Fig.~\ref{fig:phase}(a), while the key features of the bulk dynamics in the three phases are illustrated in Figs.~\ref{fig:phase}(b)-\ref{fig:phase}(d).  
In Table \ref{tab:1}, we list some of their defining characteristics.  Since these three phases all have the same average $\bar\rho$, they cannot be distinguished by any simple time-averaged measurements.  But they do differ qualitatively in their $\mathbb{P}(\rho_{\rm NEAS})$, which can be seen 
by the scaling of the instantaneous total entropy and mutual information of the NEAS.  Alternatively, time-resolved measurements of two-point functions will generically distinguish these three phases.  We outline an interferometric approach in Sec.~\ref{sec:exp} that can be used to directly probe the nontrivial correlations that contribute to the mutual information.

These phases are further distinguished by the rate of entropy production in the reservoirs following a quench.  If the system starts, for example, in a pure product state with zero entropy, the initial entropy production due to coupling to the reservoirs is diffusive ($\sim \sqrt{t}$) in phases I and II, while it is ballistic $(\sim t)$ in phase III.  But if the initial density profile is different from that of the NEAS, the final entropy production in phase III is diffusive, as this profile diffusively approaches that of the NEAS.

\subsection{Operator spreading and emergence and violation of local equilibrium}
\label{sec:lrc}

Before describing our derivation of the phase diagram, we first give an ``informal'' general picture for the quantum  dynamics in this model.  We then provide a heuristic description of the entanglement structure and the emergence or not of local equilibrium in each of the three phases.

  To gain some intuitive understanding of the action of the three gate sets introduced in Eq.~(\ref{eqn:dmu})  on the qubits, it is more convenient to work in a fermion representation of the spins after a Jordan-Wigner transformation 
\be \label{eqn:jw}
  c_j = e^{i \pi \sum_{m=0}^{j-1} u_j} \sigma_j^{-}, ~ c_j^\dag = e^{i \pi \sum_{m=0}^{j-1} u_j} \sigma_j^{+}, ~
  \ee
  where $c_j~(c_j^\dag)$ is a fermionic annihilation (creation) operator acting on site $j$.
    In this representation, the gates in $\mu_0$ only induce discrete hopping of the fermions between sites, which leads to diffusive transport of the fermion density (i.e., magnetization) since the position of the gate is also chosen randomly.  In $\mu_1$, we allow ``partial swaps'' of the qubits, which, in the fermion representation, can break up local fermion density operators $n_i=c_i^\dagger c_i$ into creation and annihilation operators that act on different sites, but we forbid gates that induce interactions between the fermions.  
In $\mu_2$, the gates are allowed to induce random phases on each state of the fermion occupations, which generates interactions between the fermions, but there are no partial swaps in $\mu_2$.  These gates allow the operators $c_i^\dag$ and $c_i$ to generate local density operators $n_i=c_i^\dag c_i$ by mapping, for example, $c_i \to c_i n_{i+1}$, but they do not  break up the local  density operators.  By combining gates from $\mu_1$ and $\mu_2$, one can generate any two-qubit unitary that conserves the total fermion number.  

One of our motivations for distributing the gates according to Eq.~(\ref{eqn:dmu}) is based on the fact that ballistic operator spreading, which is associated with fast scrambling and quantum chaos, only emerges in a two-step process that requires both $p_1$ and $p_2$ to be nonzero.  The general picture for operator spreading in high-temperature quantum chaotic spin models without a conservation law was developed in Refs.~\cite{Nahum16,Nahum17,vonKeyserlingk17}.  More recently, these results were generalized to the case with a conservation law \cite{Khemani17,Tibor17}, where it was found that during the evolution induced by the random circuit, an initially nontrivial local operator $\mathcal{O}_i$ at site $i$ can  be decomposed into two components: a conserved component  that spreads diffusively with the diffusion constant $D$, and a nonconserved component distributed across an exponentially growing number of operator strings. Each operator string has a maximum length and number of nontrivial operators that scales as $v_B t$, where $v_B$ is the butterfly velocity.  Stated more precisely, we decompose the time-evolved operator into a complete operator basis formed by tensor product ``strings'' of operators from the set $\{ \mathbb{I}, n, c, c^\dag \}$
\be
\mathcal{O}_i(t) =\mathcal{O}_{i}^c (t)+\mathcal{O}_{i}^{\rm nc}(t) = \sum_{S_c} a_{S_c} S_c + \sum_{S_{\rm nc}} a_{S_{\rm nc}}S_{\rm nc},
\ee
where $S_c$ are operator strings composed of tensor products of $\mathbb{I}$ and $n$ operators only and $S_{\rm nc}$ consists of all other operator strings.  The ballistic spreading of the front is determined by the dynamics of $\mathcal{O}_i^{\rm nc}$.    To describe the qualitative features of the operator-spreading process, we first define a coarse-grained density of local density and creation operators 
\begin{align}
n_o(x,t) &=\frac{1}{\Delta}\sum_{|y-x|<\Delta} \sum_{\{ S_{\rm nc} :\, S_{\rm nc}^y = n  \}} |a_{S_{\rm nc}}|^2, \\
c_o(x,t) &= \frac{1}{\Delta} \sum_{|y-x|<\Delta} \sum_{\{ S_{\rm nc} :\, S_{\rm nc}^y = c \}} |a_{S_{\rm nc}}|^2, 
\end{align}
where $\Delta \gg 1$ is the coarse-graining scale.
At the ``front'' of the spreading operator, each of these components is at a low density, which allows us to neglect nonlinearities in their dynamics.  The linear hydrodynamics for these two fields has to take the form
\begin{align} \label{eqn:no}
\frac{d n_o}{dt}&= {D} \frac{d^2 n_o}{d x^2} - r_1 n_o + 2 r_2  c_o,\\
\frac{d c_o}{dt}& = {D} \frac{d^2 c_o}{dx^2} +  r_1 n_o,  \label{eqn:c}
\end{align}
where $r_1 \sim p_1$ and $r_2 \sim p_2$ are the rates for generating creation and local density operators, respectively.   These equations describe a runaway process whereby the noninteracting fermion gates from $\mu_1$ break up density operators into creation and annihilation operators, which then allows the generation of more density operators through the application of the interaction gates from $\mu_2$.   Because of the constant application of swap gates, the diffusion constant is always order one  in these random circuits, but, for small $p_1$ and $p_2$, Eqs.~(\ref{eqn:no})and (\ref{eqn:c}) predict the scaling of the butterfly velocity as
\be \label{eqn:vb}
v_B^2 \sim {D} \min(\sqrt{p_1 p_2},p_2).
\ee
 The asymmetry between $p_1$ and $p_2$ arises from the fact that the partial swaps in $\mu_1$ only move operators but do not produce new operators, while the interactions in $\mu_2$ can make new density operators. To derive the scaling in (\ref{eqn:vb}), we thus have to consider the limits of small $p_1$ and $p_2$ separately.  
For $p_2 \ll p_1$, the front is a region that has diffusing $c$ and $c^{\dag}$ operators, but an underpopulation of $n$ operators.  By diffusion, if the front is moving with $v_B$, the width of the front is about ${D}/v_B$.  Thus, the total net rate of $n$ production over the front  is about $p_2 {D}/v_B$  and this has to supply the needed $n$ operators to advance the front, implying $p_2 {D}/v_B  \sim v_B$  or  $v_B^2 \sim D p_2$.  In the opposite limit of $p_1 \ll p_2$,  the front is a region where the $n$ operators are diffusing, but there is an underpopulation of $c$ and $c^\dag$ operators.   The $p_2$ process uses these $c$ operators to make the $n$ operators at the needed rate, which  gives  $p_2 c_o {D}/v_B  \sim v_B$.  The $p_1$ process breaks up the $n$ operators into $c$ and $c^\dag$ operators in this region at the rate needed to advance the front, 
which gives  $p_1 {D} v_B \sim c_o\, v_B$.  Dividing these two equations, we get  $c_o^2 \sim p_1/p_2$  and $ v_B^4 \sim  {D}^2 p_1 p_2$. 
For either $p_1$ or $p_2$ equal to zero, the butterfly velocity is exactly zero, which is the origin of the two distinct nonchaotic phases described in the introduction.

\begin{figure}[tb]
\begin{center}
\includegraphics[width=0.49 \textwidth]{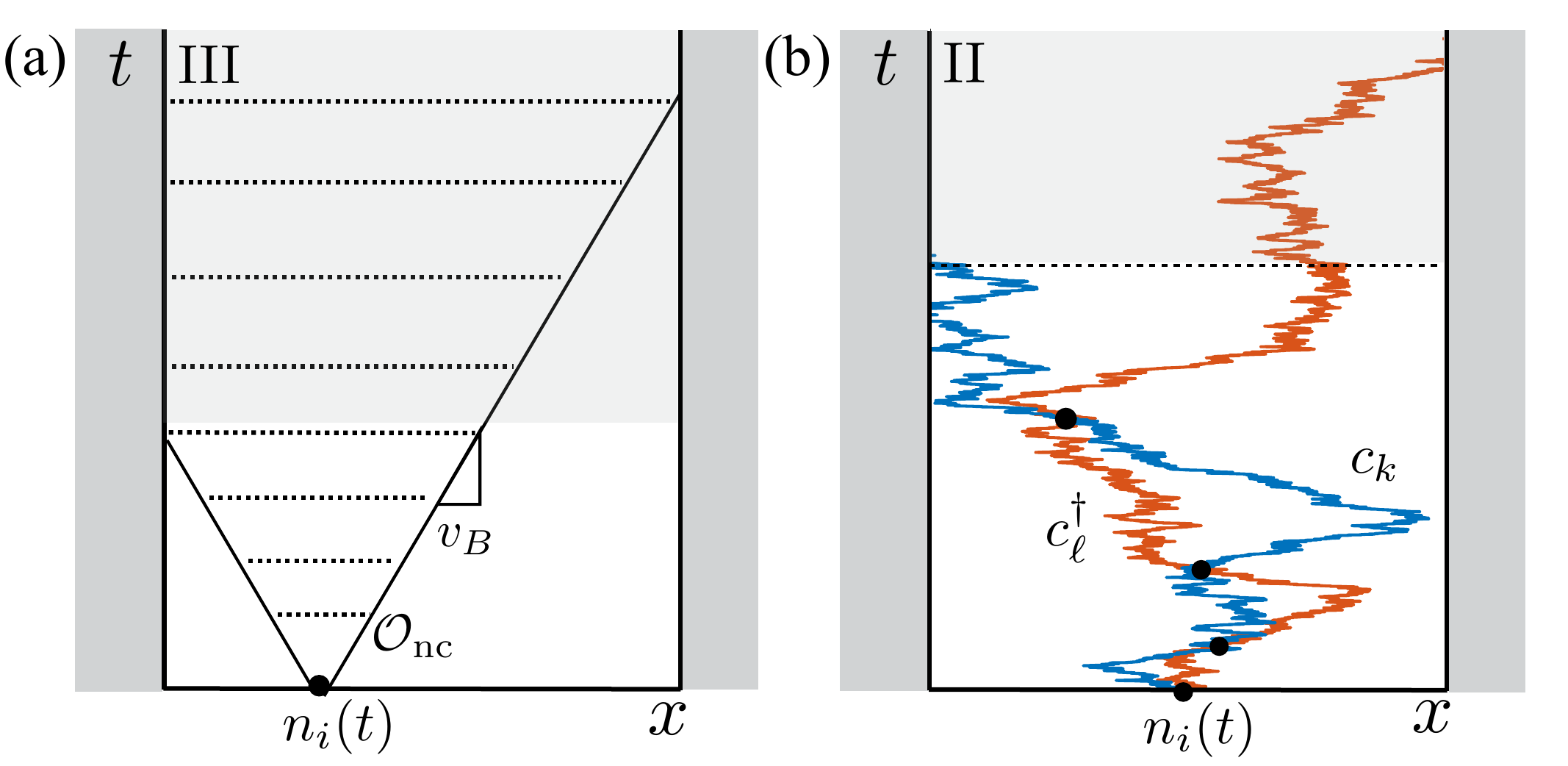}
\caption{Operator dynamics in the open system.  (a) In the quantum chaotic phase III, initially local operators generate highly entangled non-conserved operators $\mathcal{O}_{\rm nc}$ that spread ballistically at the butterfly velocity $v_B$.  These correlations are encoded in  extensive  operator strings, which are decohered at the boundary in a time of about $ L/v_B$.  Such processes only contribute a subextensive amount to the mutual information and entanglement in $\rho_{\rm NEAS}(t)$ because their production rate, given by an Ohm's law of dissipation $J^2 \sim D \delta^2/L^2$, is not enough to compensate for their short lifetime. (b) In the noninteracting fermion phase II, an initially local density operator can be split into a nonlocal pair of fermion operators.   This pair lives for a diffusive time in the system $\tau_{\rm Th} =L^2/D$ before one of these operators reaches the boundary, where it becomes decohered by the reservoir.  Such nonlocal correlations give rise to the extensive mutual information and entanglement in $\rho_{\rm NEAS}(t)$ since their production rate $J^2 \sim D \delta^2/L^2$ is sufficiently fast compared to their lifetime $\tau_{\rm Th}$.   }
\label{fig:operdyn}
\end{center}
\end{figure}

We now move to the description of the long-time behavior of the current-driven problem with open boundaries.  The crucial feature of the quantum chaotic phase for finite $p_1$ and $p_2$ (phase III) is the presence of ballistic operator spreading, whereas the operator spreading is diffusive in phases I and II, where $p_1=0$ and $p_2=0$, respectively.  This separation of timescales between the spreading of correlations and the diffusion of the conserved charge leads to a simple heuristic picture for the emergence of local equilibrium, which is illustrated in Figs.~\ref{fig:operdyn}(a)-\ref{fig:operdyn}(b). Because of the action of the partial swaps in phases II and III, the system's NEAS density matrix 
is constantly ``emitting'' nonconserved operators $\mathcal{O}_{\rm nc}$ at a rate per site that scales as the square of the current $J^2 \sim D \delta^2/L^2$ (where $\delta$ is the end-to-end difference of the  local magnetization, and $L$ is the length of the system). 
In phase III, these nonconserved operators spread ballistically at speed $v_B$ and thus only ``live'' for a time of about $ L/v_B$ before they reach a boundary and are ``absorbed'' (decohered) by that reservoir [see Fig.~\ref{fig:operdyn}(a)], implying that they can only accumulate to a density on the order of $D \delta^2/Lv_B$.  Since these nonconserved operators are needed to preserve the unitary evolution in the bulk, when they are quickly lost to the reservoir and remain at a low density, the deviation $\Delta S = S(\rho_{\rm LE})-\overline{S(\rho_{\rm NEAS})}$ remains small.   

Understanding the volume-law deviation of the entropy from local equilibrium in the other two phases requires separate consideration:  In phase I, the argument above no longer applies because the absence of partial swaps implies that the NEAS has a zero rate for ``emitting'' nonconserved operators.  In this case, the extensive deviation from local equilibrium is encoded in the single-site spin densities: In the NEAS, each site came most recently from one of the two reservoir,s and it still has the same spin density.  In phase II, the ``emission'' process is present, but the nonconserved operators spread only diffusively; thus, they live for a time of about $L^2/D$ and build up to a constant density in the thermodynamic limit on the order of $\delta^2$ [see Fig.~\ref{fig:operdyn}(b)].  We find that this finite density of non-conserved operators leads to a volume-law deviation of the entropy, mutual information, and entanglement away from local equilibrium.  

\subsection{Phase diagram}

In this section, we describe our derivation of the phase diagram for the NEAS as a function of $p_1$ and $p_2$ with many of the technical details underlying this analysis provided in the appendixes.

\subsubsection{Phase I: Discrete hopping limit} 

For $p_1=0$, we can see, by inspection, that diagonal product states of the form
\be
\rho_{\bm \tau} = \bigotimes_{i=1}^{L}[ \tau_i \rho_{m_L} + (1- \tau_i) \rho_{m_R}]
\ee
are NEASs of the random circuit, where $\tau_i$ is a pseudospin variable that keeps track of whether a given density operator was originally inserted from the left or right reservoir.  The attracting nature of these states arises because the random phases for the gates in (\ref{eqn:p2gates}) do not affect this type of product state, while the SWAPs simply rearrange the configuration of on-site density matrices.  If the initial state has any off-diagonal coherences in this $u$, $d$ basis, they will diffuse to the boundaries and ``disappear'' into the reservoirs.  The dynamics within the NEAS manifold can then be mapped to the SSEP where  $\rho_{m_L/m_R}$ maps to a pseudospin-up or pseudospin-down state at a given site and the pseudospin reservoirs are fully polarized.  This result allows us to characterize the entire distribution function of the NEASs through the relation $\mathbb{P}(\rho_{\bm \tau}) = P(\bm{\tau})$ for $P(\bm{\tau})$ satisfying Eq.~(\ref{eqn:dpdt}) with $\delta = 1$.  Note that this case is perhaps an unconventional perspective on the classical SSEP, which is typically formulated as an intrinsically stochastic and dissipative process. Instead, we view the dynamics as produced by one particular circuit, and for finite $L$ for that specific circuit there is a unique time-dependent absorbing state $\rho_{\bm \tau}(t)$.  

We can find the average entropy of the NEASs by noting that 
\be \label{eqn:srhoi}
S(\rho_{\bm \tau}) = N_L S(\rho_{m_L}) + N_R S(\rho_{m_R}),
\ee
where $N_L$ is the total number of pseudospin-up states of the chain and $N_R=L-N_L$ is the number of pseudospin-down states.  Because of the pseudospin $z$ inversion symmetry, $\overline{N_L}=\overline{N_R} = L/2$.  Considering the antisymmetric case $m_L = -m_R=\delta/2$ and comparing to local equilibrium, we find,  after averaging Eq.~(\ref{eqn:srhoi}) over circuits,
the volume-law correction
\be \label{eqn:delSssep}
\begin{split}
\Delta S & =  S(\rho_{\rm LE}) - \overline{S(\rho_{\bm \tau})} \\
& =\frac{L}{2} - \frac{L}{2 \delta} (1-\delta^2)\tanh^{-1}(\delta) + O(L^0).
\end{split}
\ee
In the limit $\delta \to 1$, the NEAS approaches a pure state, and this deviation reaches its maximum possible value.

\subsubsection{Phase II: Diffusive noninteracting fermions}
\label{sec:phaseII}

  For $p_2=0$,  the NEASs have exact representations as Gaussian fermionic states because the dynamics are equivalent to that of noninteracting fermions and, in the fermion representation, the reservoirs are clearly Gaussian states.  Such states are uniquely determined by their two-point function \cite{Chung01,Cheong04,Peschel09}:
\begin{align}
G_{ij}& = \trace[ \rho \,c_i^\dagger c_j ],  \\ \label{eqn:rhoNIF}
S(\rho) & = - \trace [(\mathbb{I}-G) \log (\mathbb{I}-G) ] - \trace [G \log G],
\end{align}
where $c_i$ are fermion annihilation operators [see (\ref{eqn:jw})].  The operators $c_i$ spread diffusively in this random circuit, leading to diffusive spreading of both the spin density $n_i = c_{i}^{\dagger}c_i$ and the coherences $c_{i}^{\dagger}c_j$ for $i\ne j$ [see Fig.~\ref{fig:phase}c].  Again, any non-Gaussian features of the initial state will diffuse to the boundaries and disappear into the reservoirs, leaving the Gaussian NEAS at long times.
The presence of the log of the two-point function makes it difficult to compute the average entropy.  For simplicity, we restrict ourselves to anti-symmetric reservoirs $m_L=-m_R=\delta/2$ and small $\delta$, which allows us to expand the $\log$ by transforming into an eigenbasis of $G$ at each instance of time
\be \label{eqn:entNIF}
\begin{split}
\overline{ S(\rho_{\rm NEAS}) } & \approx -2  \trace [ \bar{G} \log \bar{G} ] -2 \trace [ \overline{ \delta G^2}]  ,
\end{split}
\ee
where  $\delta G = G - \bar{G}$.  Since $\bar{G}$ is just given by the linear magnetization profile, the first term is the entropy of local equilibrium, while the second term accounts for the deviations arising from additional correlations and is determined by the covariance matrix 
\be \label{eqn:cov}
A_{ij} \equiv [\overline{ \delta G^2}]_{ij}= (1-\delta_{ij})  \overline{\lvert \langle c_i^\dag c_j \rangle \lvert^2 }  + \delta_{ij} \big( \overline{\mean{n_i}^2 } - \overline{\mean{n_i}}^2\big).
\ee
To solve for $A_{ij}$, we work in the scaling limit $(L\to \infty)$ for $p_1L^2 \gg 1$ and only compute the lowest-order correction in a $1/L$ expansion.  Defining the coordinates $ x = i/L$ and $y=j/L$, we introduce the variables $a(x,y)=A_{xL,yL+1}$ and $h(x)=A_{xL,xL}$.  The restriction to nearest-neighbor gates implies that, away from the diagonal $x =y$, $a(x,y)$ satisfies a diffusion equation with boundary conditions $a(x,1)=a(0,y)=0$.  Integrating out the $h(x)$ variable, one finds that $a(x,y)$ has a constant source term along the diagonal given by $-2 J^2 L \delta(x-y)= -2 \delta^2\delta(x-y)/L$.
This diffusion problem has the solution (for $x<y$)
\be \label{eqn:axy}
a(x,y) = \frac{ x (1-y)}{L} \delta^2 +O(L^{-2}) +  O (\delta^4)~.
\ee
The deviation from local equilibrium can be expressed perturbatively in $\delta$ as
 \be
 \begin{split}
\Delta S&\approx   4 L^2  \int_0^1 dy \int _0^y dx \, a(x,y) =  \frac{\delta^2}{6}L~.
 \end{split}
 \ee
 which has a volume-law correction away from both local equilibrium and the average entropy of phase I.  We can also compute the mutual information of two sections of the chain cut at a point $z\in(0,1)$,
 \begin{align}
\overline{I(L:R)}& =   \delta^2 z^2(1-z)^2 L + O(L^0)+O(\delta^4)~.
 \end{align}
In sharp contrast to the $p_1=0$ solution, we find that the NEASs have a volume-law scaling of the mutual information, indicating that these states are highly correlated.  

For mixed states, the mutual information is not a direct measure of entanglement as it can be dominated by classical correlations \cite{Eisert10}.  In Sec.~\ref{sec:entphase} we explicitly show that the NEAS density matrix has volume-law entanglement  for sufficiently large $\delta$ by computing a lower bound on the logarithmic negativity.   When measured according to the fermionic logarithmic negativity recently introduced by Shapourian, Shiozaki, and Ryu \cite{Shapourian17}, we find that this volume-law scaling persists down to arbitrarily small $\delta$.  This result establishes that the NEAS density matrix is driven to a nonseparable state in the large-$L$ limit. 
Thus, the volume-law deviation of the entropy in phase II arises from  entanglement and nonlocal correlations, while the deviation in phase I arises from ``classical,'' single-site magnetizations.  

\subsubsection{Phase III: Quantum chaotic}

  For nonzero values of $p_1$ and $p_2$, the dynamics in the bulk are quantum chaotic as discussed in Sec.~\ref{sec:lrc}.  Recall that the effective butterfly velocity $v_B$, which measures the speed of the ballistic operator front, scales as $v_B^2 \sim \min(\sqrt{p_1 p_2},p_2)$, while the diffusion constant for conserved quantities is identical for all values of $p_1$ and $p_2$.  The calculation of the average entropy in this region is more difficult than in phases I and II because this random circuit does not map to an integrable model.  To approach the calculation analytically, we instead work perturbatively in $\delta$.  This method allows us to expand $ \log \rho_{\rm NEAS} = \log(\bar{\rho}+\rho_{\rm NEAS}-\bar{\rho})$ to derive an expression for the average entropy similar to Eq.~(\ref{eqn:entNIF}),
\be \label{eqn:avgent}
\begin{split}
\Delta S & \approx   2^{L-1} \big( \trace[ \overline{\rho_{\rm NEAS}^2}] - \trace[ \bar{\rho}^2] \big),
\end{split}
\ee
which reduces the computation of the average entropy to the easier task of computing the average purity.  The details of this calculation are described in Appendix \ref{app:phaseIII}, while we give a brief summary here, with more details provided in the discussion of the crossover behavior in Sec.~\ref{sec:crossovers}.
The approach we take is to first derive exact steady-state equations for the average replicated density matrix $\overline{\rho_{\rm NEAS}\otimes \rho_{\rm NEAS}}$.  The  solution to these equations can then be mapped to the NESS of a fictitious six-species stochastic lattice gas model, which we refer to as the $abc$ model.  Here, the $a$ particles represent the off-diagonal coherences of the density matrix, while the $b$ and $c$ particles represent different types of density-density correlations.  To solve this model, we use the ansatz that the $n$-point connected correlation functions of ${\overline{\rho_{\rm NEAS} \otimes \rho_{\rm NEAS}}}$ have a scaling as $\delta^n/L^{n-1}$ or smaller.  This type of scaling is well known for the SSEP and can be proved exactly for phases I and II.   
 Using this ansatz, for $ p_{1,2} L^2 \gg 1$, the lowest-order correction to the entropy deviation from local equilibrium is given by
 \be \label{eqn:delsiii}
 \Delta S =\frac{\alpha_1 \delta^2}{L p_1} + \frac{\alpha_2 \delta^2}{L p_2}  +O(\delta^4) + O(L^{-2}),
 \ee
 where $\alpha_{1}=(3-p_1)$, $\alpha_2=2+\sqrt{2 p_2/(2-p_2)}$ are order-one constants computed from the perturbative solution to the the $abc$ model (see Appendix \ref{app:phaseIII}).  As a result, the average entropy is given by that of local equilibrium in the large-$L$ limit.  This behavior is similar to the SSEP, where the leading-order $L^0$ term in $\Delta S$ at small $\delta$ scales as $\delta^4$.  We can bound the higher-order corrections to $\Delta S$ since $S(\rho_{\rm LE}) \ge S(\bar{\rho}) \ge \overline{S(\rho_{\rm NEAS})}$ and $S(\bar{\rho})$ is equal to $S(\rho_{\rm LE})$ up to a constant correction \cite{Bahadoran07,Derrida06}.  
 To second order in $\delta$ and for $ p_{1,2} L^2 \gg 1$, we  find that the mutual information  between two halves of the chain has the scaling
 \be \label{eqn:mutinfiii}
 \overline{I(L:R)} = \frac{\alpha_3 \delta^2}{L^2 p_2^{3/2}} + O(\delta^4)+O(L^{-3}),
 \ee
 where $\alpha_3 = \sqrt{2(2-p_2)}$ is computed from the $abc$ model. 
Based on our analysis, we suspect that $\mathbb{P}(\rho_{\rm NEAS})$ is concentrated near $\bar{\rho}$, whose mutual information is bounded by $\log L$ for large $\delta$, but further analysis is required to make more definitive statements about the mutual information at large values of $\delta$.

\subsection{Crossovers between phases}
\label{sec:crossovers}

Here, we carefully analyze the finite-size scaling of the crossovers between each phase to gain further insight into the relevant length scales that describe the $\textrm{NEASs}$.  Conveniently, we have an analytic description of the entire phase diagram for this model using the perturbative solution to the  $abc$ model for $\overline{\rho_{\rm NEAS}\otimes\rho_{\rm NEAS}}$ introduced in Appendix \ref{app:phaseIII}.  Essentially, to second order in $\delta$, the average replicated density matrix for antisymmetric reservoirs ($m_L=-m_R$) is entirely determined by six correlation functions: the average magnetization profile $\overline{\mean{u_i}}$, the average connected spin-spin correlations $\tau_{ij}\equiv \overline{\mean{u_i u_j}} - \overline{\mean{u_i}} \overline{\mean{u_j}}$, and the average fluctuations 
\begin{align}
h_i &\equiv \overline{\mean{u_i}^2} - \overline{\mean{u_i}}^2 ,\\
b_{ij}&\equiv \mean{\mean{u_i u_j}^2}_c + \mean{\mean{u_i d_j}^2}_c,\\
B_{ij}&\equiv \mean{\mean{u_i u_j}^2}_c - \mean{\mean{u_i d_j}^2}_c, \\
a_{ij} & \equiv \overline{\lvert \mean{\sigma_i^{+} \sigma_j^{-}}\lvert^2},
\end{align}
where $\mean{\mean{AB}^2}_c\equiv \overline{\mean{AB}^2}- \overline{\mean{A}^2} \overline{\mean{B}^2}$ denote the connected correlations defined with respect to the single-site fluctuations.  For antisymmetric reservoirs, these correlations are invariant under $\delta \to -\delta$, which sends $u_i \to d_i$ and vice versa.  The average connected spin-spin correlation is known exactly from the mapping of $\bar{\rho}$ to the NESS of the SSEP and is given by (see Appendix \ref{app:SSEP}):
\be
\tau_{ij} = - \frac{i (L+1-j) }{L (L+1)^2}\delta^2.
\ee
  We define the scaled variables $x\equiv i/L$ and $y\equiv (j-1)/L$, then, expanding to lowest order in $1/L$ and second order in $\delta$, we find that these correlation functions satisfy a set of diffusion equations away from the diagonal $x=y$,
\begin{align} \label{eqn:hxmt}
h''(x)&=  \frac{4 p_1 L^2  }{3- p_1} \delta J^2(x) - 2 \delta^2, \\ \label{eqn:axymt}
 \nabla^2 a(x,y) &= \frac{4 p_2 L^2}{2-p_2} a(x,y), \\
 ~\nabla^2 b(x,y)&= \frac{4 p_1 L^2}{3- p_1} b(x,y), ~\nabla^2 B(x,y)= 0,
\end{align}
where $\delta J^2(x) \equiv h(x) +B(x,x) - a(x,x)-\tau(x,x)$ gives the corrections to the current  fluctuations as compared to local equilibrium.  In terms of these variables, we can compute the entropy and mutual information as
\begin{align}
\Delta S&  = 2L \int_0^1 dx h(x) \\ \nonumber
&+ 4L^2 \int_0^1 dy\int_0^y dx \,  [a(x,y) +2 b(x,y)], \\ 
\overline{I(L:R)} & = 4 L^2 \int_{z}^{1} dy \int_{0}^{z} dx\,  [a(x,y)+2 b(x,y)].
\end{align}
The boundary conditions are such that all the connected correlation functions vanish for $x$ or $y$ equal to $0$ or $1$.  There is also a set of boundary conditions on the diagonal, which mix these different functions and result in a nontrivial steady-state.  The boundary condition for the coherences takes the form
\be \label{eqn:aboundc}
(\del_x - \del_y) \, a(x,y)\lvert_{x=y} \, =  \frac{4 p_1 L}{6-3 p_2}\Big[\frac{ \delta^2}{2 L^2}  +  \delta J^2(x) \Big],
\ee
which provides a microscopic justification for the heuristic picture outlined in Sec.~\ref{sec:lrc}, where we argued that the action of the partial swaps leads to nonconserved operators being generated at a rate per site proportional to the average current fluctuations.  From  Eq.~(\ref{eqn:aboundc}), we find that there is an additional source term $\delta J^2(x)$, which accounts for the deviations of the current fluctuations from local equilibrium and has to be found self-consistently. 
  The full expressions for the diagonal boundary conditions for $b(x,y)$ and $B(x,y)$ are given in Appendix \ref{app:phaseIII}.  Solving these equations we find that $b(x,y)$ is only nonzero at order $\delta^4$ for all values of $p_1$ and $p_2$.  Our expectation from the SSEP is that the $\delta^4$ corrections  scale with a higher power of $L$, which would imply that these density-density fluctuations only give rise to an area-law correction to the entropy and mutual information.  The field $B(x,y)\sim \delta^2/L$ has a nontrivial dependence on $p_1$ and $p_2$, but this field gives no direct contribution to the entropy at second order in $\delta$.

Using these coupled diffusion equations, we can obtain a quantitative picture for the crossover behavior.  The results are summarized in Fig.~\ref{fig:cross}.  In panel (a), we identify the  crossover boundaries between the three phases as determined by the scaling of $\Delta S$ and $\overline{I(L:R)}$, while panel (b) schematically shows the functional behavior of both quantities in the crossover from (A) phase I to II, (B) phase II to III, and (C) phase III to I. 

\begin{figure}[tb]
\begin{center}
\includegraphics[width = .47 \textwidth]{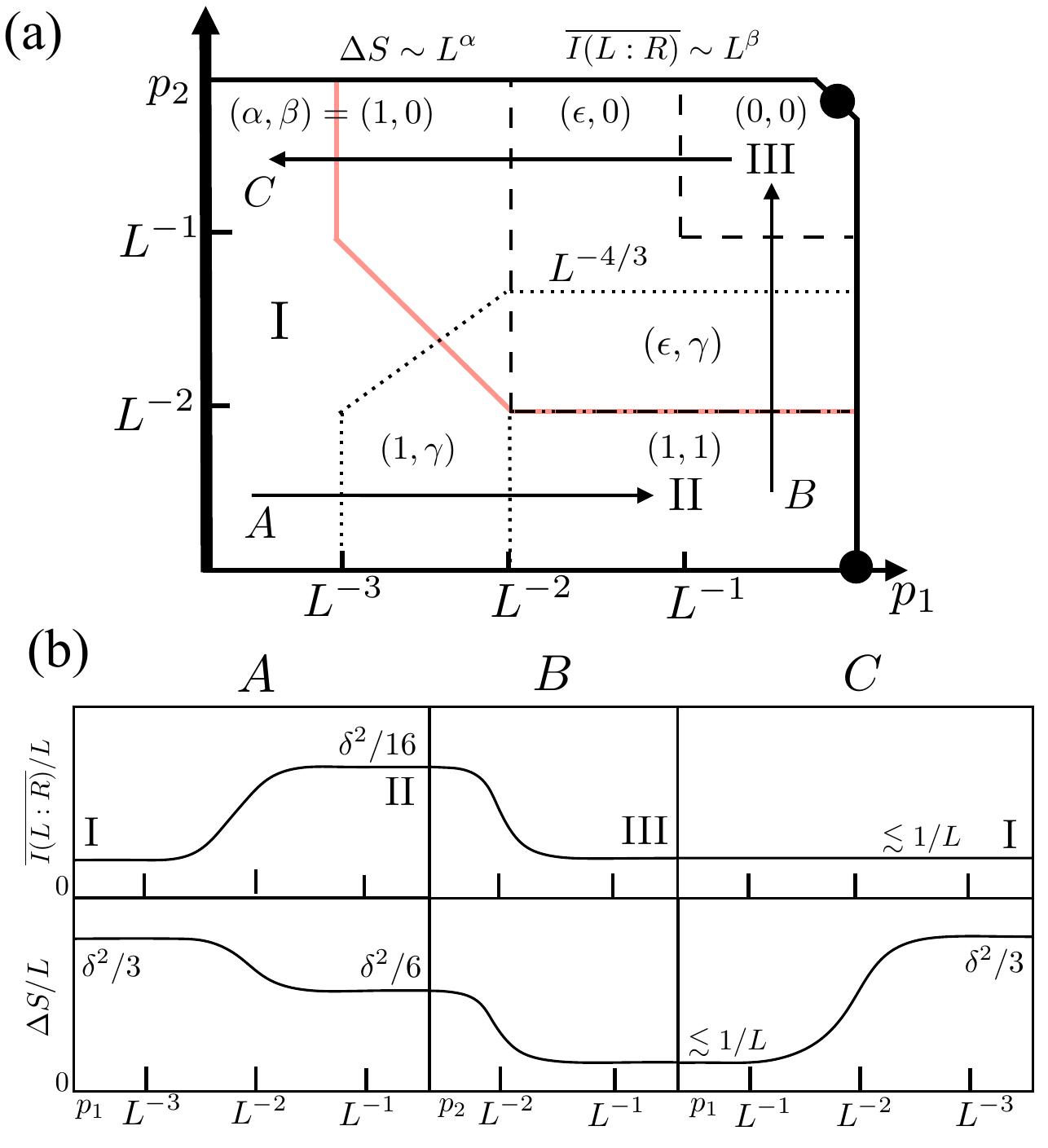}
\caption{(a) Crossover boundaries for the NEAS phase diagram for small $\delta$.  Dotted lines denote mutual information crossovers, dashed lines denote crossovers in the von Neumann entropy, and the pink boundary denotes the crossover from diffusive to ballistic coupling to the reservoirs.  Here, $0 < (\alpha,\beta) < 1$ are the scaling exponents for $\Delta S$ and $\overline{I(L:R)}$, respectively.  The exponents $\epsilon$ and $\gamma$  denote the crossover values for the scaling exponents and take continuous values between 0 and 1.  The $L^{-4/3}$ scaling for the mutual information arises because the long-range correlations that contribute to the mutual information are more strongly suppressed than the single-site density fluctuations that only contribute to $\Delta S$.  (b) Schematic behavior of the entropy deviation density and the mutual information density for the three cuts in parameter space labeled in panel (a).}
\label{fig:cross}
\end{center}
\end{figure}

\emph{Phase I to II crossover.---}The phase I to II crossover is distinct from those to phase III because all operators spread diffusively in both phases so ballistic operator spreading does not play a role.  Instead, one can understand the crossover solely in terms of the dynamics of diffusive noninteracting fermions.  In phase I, the fermion operators that have an expectation value in the NEAS are predominantly paired into density operators $n_i = c_i^\dagger c_i$.  Once injected into the system from one of the reservoirs, these operators live for a time on the order of the Thouless time $\tau_{\rm Th}= L^2/D$ before escaping back out to the reservoirs.  However, if a  partial swap gate acts on this operator within $\tau_{\rm Th}$, then $n_i$ will break up into a pair of creation and annihilation operators.  These operators then evolve independently in the system before escaping to the reservoir.  Their common origin from the initially local density operator gives rise to the transient long-range entanglement in the system.   In phase II, this is in fact the dominating physics as the density operators are  almost immediately broken up into  pairs of independent creation and annihilation operators once they enter the system.  Based on this argument, we expect the crossover to occur for $p_1 L^2 \sim 1$.   

We can verify this scaling directly using  Eqs.~(\ref{eqn:hxmt}) and (\ref{eqn:aboundc}).  For $p_1 L^2 \ll 1$, Eq.~(\ref{eqn:hxmt}) has the approximate solution
\be
h(x) \approx  x(1-x) \delta^2, 
\ee
which, unlike local equilibrium, leads to order $L^0$ current fluctuations $\delta J^2(x)$ instead of $L^{-2}$.  Thus, despite the small value of $p_1$, these large current fluctuations allow the partial swaps to drive coherences on the diagonal at the rate per site of about $p_1$,
\be
(\del_x - \del_y) \, a(x,y)\lvert_{x=y} \, =  \frac{4 p_1 L}{6-3 p_2} \delta^2 x(1-x).
\ee
This diagonal boundary condition leads to the scaling $a(x,y) \sim p_1 L $ and $\overline{I(L:R)} = \beta(p_1,p_2) p_1 \delta^2 L^3 $ for an order-one constant $\beta(p_1,p_2)$, which implies the emergence of the volume-law mutual information as $p_1L^2$ approaches 1.  Remaining in phase I, but increasing $p_2 L^2$ above 1, this scaling becomes instead $\overline{I(L:R)}=\beta(p_1,p_2) p_1 \delta^2/p_2^{3/2}$, which leads to the black-diagonal crossover boundary for the mutual information shown in Fig.~\ref{fig:cross}(a).

We now return  to the description of the phase I to II crossover on the phase II side ($p_1 L^2 \gg 1$). Unlike the phase I side, the current fluctuations in this regime are close to their value in local equilibrium.  Away from order $1/\sqrt{p_1}$ sites near the boundary, they are given by
\be
\delta J^2(x) \approx \frac{3 \delta^2}{2 p_1 L^2} - \frac{\delta^2}{2 L^2}.
\ee
Inserting this into Eq.~(\ref{eqn:aboundc}) implies that the diagonal source term for $a(x,y)$ becomes independent of $p_1$, which reflects the fact that the density operators are quickly broken up by the partial swaps within the Thouless time.  Although the current fluctuations are weak in this regime, the rapid application of the $p_1$ gates allows the coherences to build up to the large density needed to recover the volume-law mutual information.

\emph{Phase II to III crossover.---}Because of its broader physical significance to disordered metallic systems, the physical picture underlying the phase II to III crossover is discussed in more detail in Sec.~\ref{sec:appLE}.  Since the butterfly velocity scales as $\sqrt{p_2}$ in this regime, there is a natural  crossover scale defined by the value of $p_2$ for which $L/v_B \sim \tau_{\rm Th}$, which suggests that the crossover occurs when $p_2 L^2 \sim 1$.  This is readily verified using Eq.~(\ref{eqn:axymt}), which shows that when $p_2 L^2 \gg 1$, the $a(x,y)$ is exponentially damped away from the diagonal.  The entropy deviation and mutual information in this regime are given by Eqs.~(\ref{eqn:delsiii}) and (\ref{eqn:mutinfiii}), which  determine the boundaries shown in Fig.~\ref{fig:cross}(a).  

\emph{Phase I to III crossover.---}The phase I to III crossover, which occurs in the regime $p_1<p_2$, is distinct from the other two crossovers because the system can no longer be treated as noninteracting or weakly interacting fermions in the crossover regime.  Moreover, because of the weak scaling of the butterfly velocity as $(p_1p_2)^{1/4}$ in this regime, the ballistic operator spreading can play a secondary role in determining the properties of the NEAS crossover.  This feature of the phase I to III crossover follows from a similar argument that was used in describing the phase I to II crossover.  Since the partial swaps are needed to break up the single-site density operators, the action of the $p_1$ gates is crucial to realizing local equilibrium; however, for $p_1L^2 \ll 1$, the probability of the partial swaps acting before the density operators are diffusively exchanged with the reservoirs is small.  The distinction from the phase I to II crossover is that the ballistic operator spreading can lead to a rapid dissipation of the nonconserved operators to the reservoirs even for $p_1 L^2 \sim 1$.  In fact, according to our expression for the butterfly velocity, this overdamping of the coherences persists until $p_1p_2 L^4 \sim 1$.   Of course for too-small values of $p_1$, our derivation of $v_B$ no longer applies because the entire system reaches the NEAS after application of about $L^3$ gates (note that this is the number of gates applied to the system within a Thouless time $\tau_{\rm Th} =  L^2/D$); however, $p_1 L^3 \sim 1$ is still a much smaller scale than the crossover scale set by the Thouless time of $p_1 L^2 \sim 1$.  

This overdamping of the nonconserved operators motivates us to consider an approximate description of the phase I to III crossover in which we model decoherence induced by the reservoir as an instantaneous measurement of the density matrix in the $z$ basis after each unitary is applied.  In this simplified model, the density matrix always remains in a diagonal mixed state such that the dynamics have an effective classical description.  More specifically, each time step of the random circuit is replaced by a randomly chosen operation on the diagonal components of the nearest-neighbor density matrices that, with probability $1-p_1$, applies either the identity or SWAP operation and, with probability $p_1$, applies a transition matrix that is the identity except when acting on the operators $\{ u_{i}d_{i+1}$,$d_{i}u_{i+1}\}$, where it takes the form
\be
\left( \begin{array}{c c c c}
 \cos^2 (\theta/2) & \sin^2 (\theta/2)  \\
 \sin^2 (\theta/2) & \cos^2 (\theta/2)  \\
\end{array}
\right),
\ee
where $\theta$ is a random variable drawn from $[0,\pi]$ with probability density $\sin \theta$.    The average replicated density matrix for this model is described by the same $abc$ model as the random circuit, with the restriction that configurations with $a$ particles are no longer allowed.  As a result, to second order in $\delta$, $\overline{\rho \otimes \rho}$ is characterized  in terms of the correlation functions $\overline{\mean{u_i}}$, $\tau_{ij}$, $h_i$, $b_{ij}$, and $B_{ij}$.  Furthermore,  these correlations satisfy the same equations of motion as the random unitary circuit with the constraint  $a(x,y)=0$.   While both models exhibit the same crossover scale at $p_1 L^2 \sim 1$, we can gain some physical intuition by considering the simpler model without coherences.  For $p_1=0$, we have  the same result as for the random unitary circuit --- that $\mathbb{P}(\rho_{\rm NEAS})$ can be found from the NESS for the SSEP.  On the other hand, for large values of $p_1$, the stochasticity induced by allowing finite $\theta$ implies that even a single realization of the  NEAS will be close to the NESS of the SSEP.  In this case, the Thouless time is the only relevant timescale in the problem since all the dynamics are diffusive, which forces the crossover scale to be $p_1 L^2 \sim 1$. 

\subsection{Volume-law entanglement in phase II}
\label{sec:entphase}

In this section, we show that the entanglement in phase II follows a volume-law scaling using our analytic solution for the covariance matrix and exact numerical simulations of the long-time evolution of the random circuit.
The von Neumann entropy and mutual information are not direct measures of entanglement for mixed-state density matrices.  Although a variety of entanglement measures for mixed states have been introduced \cite{Horodecki09}, these measures are generally harder to compute than the von Neumann entropy because of the  difficulty of distinguishing entanglement in many-body systems from nonlocal, classical correlations, which, for mixed states, can also be generated by local operations and classical communication  (LOCC) between the two regions.  One efficiently computable measure of entanglement in the Hilbert space dimension is the logarithmic negativity $\E(\rho)$, which was originally introduced by Vidal and Werner \cite{Vidal02} and later proven to be an entanglement monotone under LOCC by Plenio \cite{Plenio05}.    Note that 
$\E(\rho)$ has the property that, if it is nonzero, then the density matrix is nonseparable (i.e., entangled) between the two regions.  In bosonic or spin systems, $\E(\rho)$ is naturally defined using the partial transpose operation,
\be \label{eqn:Ta}
\bra{i_A,j_B}\rho^{T_A} \ket{\ell_A,k_B} \equiv \bra{\ell_A,j_B} \rho \ket{i_A,k_B},
\ee
where matrix elements are taken with respect to a separable orthonormal basis for regions $A$ and $B$.  The logarithmic negativity is 
\be
\E(\rho) \equiv \log  \lvert \lvert \rho^{T_A} \lvert  \lvert
\ee
where $  \lvert \lvert {A} \lvert \lvert \equiv \trace[\sqrt{A^\dag A} \, ]$ is the trace norm.  Here,$\E(\rho)$ is a measure of the strength and number of negative eigenvalues of $\rho^{T_A}$ and is an upper bound on the entanglement of distillation $\E_D(\rho)$.  Note that $\E_D(\rho)$ is a more fundamental measure of entanglement defined as the maximum number of near-perfect singlet states that can be generated from multiple copies of $\rho$ with LOCC on $A$ and $B$.  Although our underlying physical system in the random circuit model is a spin system, the logarithmic negativity for two adjacent regions in the original spin representation is equal  to the logarithmic negativity of the fermions obtained after a Jordan-Wigner transformation \cite{Eisler15}.  Since we have an exact representation of the NEAS density matrix in phase II in terms of a Gaussian fermionic state,  it is natural to ask whether we can directly compute the logarithmic negativity using this representation.

  For fermionic systems, a partial transpose operation can be defined analogously to Eq.~(\ref{eqn:Ta}); however, it has the property that the partial transpose of a fermionic Gaussian state may not be Gaussian, which makes the logarithmic negativity generally intractable to compute for large systems even for these simple fermionic states \cite{Eisler15}.  On the other hand, it was argued by Shapourian, Shiozaki, and Ryu  that a more natural definition of the logarithmic negativity for fermions  is in terms of a partial time-reversal operation on subsystem $A$  \cite{Shapourian17,Shapourian18}.  We refer to this measure as $\E_f(\rho)$ to distinguish it from $\E(\rho)$.  Unlike the partial transpose, this operation maps fermionic Gaussian states to fermionic Gaussian states, making it efficiently computable in the number of fermions.   Furthermore, for fermionic Gaussian states, it can be shown that it is an upper bound on the logarithmic negativity $\E(\rho) \le \E_f(\rho)+\log\sqrt{2}$ \cite{Eisler15,Eisert16}.  Several efficiently computable lower bounds on $\E(\rho)$ for Gaussian states were also introduced by Eisert, Eisler and Zimbor{\' a}s~\cite{Eisert16}.  Here, we use the lower bound $\E_\ell(\rho)$ introduced by these authors for Gaussian states that conserve particle number:  $\trace[\rho c_i^\dag c_j^\dag] = \trace[\rho c_i c_j] =0$ [see Sec. IVB of Ref.~\cite{Eisert16} for a definition of $\E_\ell(\rho)$].    
  
    \begin{figure}[t]
\begin{center}
\includegraphics[width = 0.47 \textwidth]{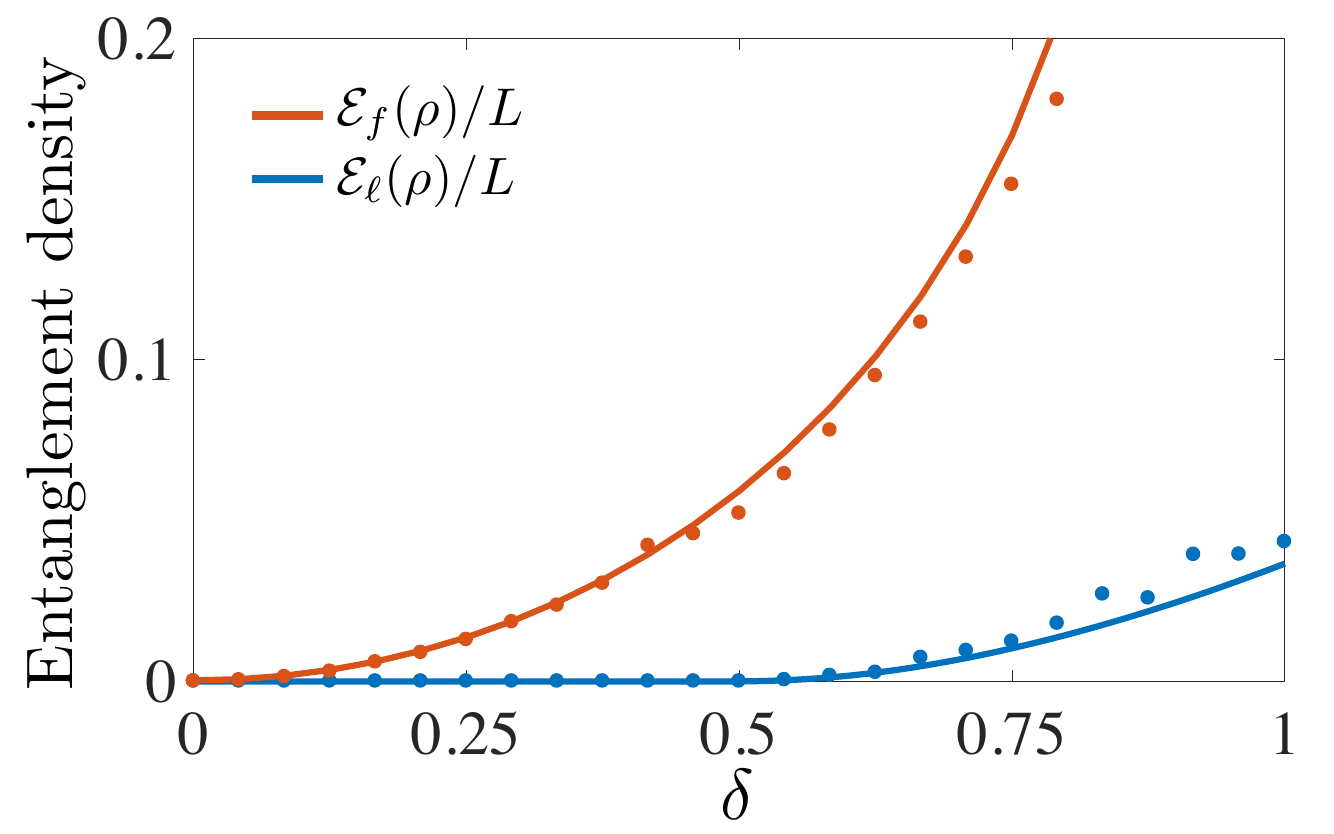}
\caption{Upper $\E_f(\rho)$ and lower $\E_\ell(\rho)$ bounds for the logarithmic negativity per unit length in phase II with $p_1=1$ and $p_2=0$.  The upper bound is equal to the fermionic logarithmic negativity \cite{Shapourian17}.  Solid lines are for $L=10^3$ sites computed by randomly sampling the fermionic two-point function of the NEAS, which is valid for small $\delta$.  Circles are computed from  an exact simulation of the NEAS two-point function for the random circuit with $L=50$ sites.  The volume-law coefficient for $\E_f(\rho)$ is nonzero for any $\delta>0$, while the volume-law coefficient of $\E_\ell(\rho)$ is nonzero only for $\delta > 0.5$.  }
\label{fig:logneg}
\end{center}
\end{figure}

  To compute the upper and lower bounds for the time-averaged logarithmic negativity in phase II, we use the fact that all fourth- and higher-order correlations of the matrix elements of the two-point function $G_{ij}$ scale as $\delta^4$ or higher.  This method allows us to sample random realizations of the NEAS two-point function by treating $G_{ij}$ as independent, normally distributed, random variables with mean $\bar{G}_{ij}$ and variance $A_{ij}$ [see Eq.~(\ref{eqn:cov})].   
  Using this approach to sample from the NEAS, we are able to compute $\E_{f,\ell}(\rho)$ for up to several thousand sites, which allows an accurate determination of the coefficient for the volume-law term in $\E_{f,\ell}(\rho)$ as a function of $\delta$.  The results are shown in Fig.~\ref{fig:logneg}, where we  compare this direct sampling approach for $L=10^3$ sites to exact simulations of the NEAS two-point function for realizations of the random circuit with $p_1=1$, $p_2=0$, and $L=50$ sites.  We find good agreement between the two independently computed coefficients up to $\delta \approx 0.75$.  Interestingly, $\E_f(\rho)$ exhibits volume-law scaling for any nonzero $\delta$, which shows that the NEAS density matrix remains nonseparable in the scaling limit.   On the other hand, $\E_\ell(\rho)$ is exactly zero up until $\delta =0.5$, at which point a clear volume-law scaling emerges.  
The sharp behavior of the lower bound suggests that the system may undergo a phase transition with increasing $\delta$ from a weakly entangled state to a volume-law entangled state. However, we expect that the actual behavior of $\E(\rho)$ is more like $\E_f(\rho)$, which exhibits a crossover to the volume-law scaling for any nonzero $\delta$.  Finally, we remark that it was recently proven that thermal mixed states of equilibrium spin models obey an area law for the logarithmic negativity \cite{Sherman16}.  Thus, similar to the observed volume-law mutual information, this volume-law scaling of $\E(\rho)$ is a manifestly nonequilibrium effect.
  
\subsection{Entropy production following a quench}

The discussion so far concerned average properties of each phase in the long-time limit.  It is also interesting to consider the explicit time dynamics of this model, where the effects of the ballistic operator spreading in the quantum chaotic region should be more manifest.  Here we make a qualitative argument that, starting from an initial state that differs extensively in entropy from the NEAS, the ballistic operator spreading directly appears in the entropy production rate in the reservoirs following a quench into the chaotic phase.

For a quench into phase I or II, all operators spread diffusively, implying that the system's net entropy change, which is only due to the reservoirs, will grow following the quench as $\sqrt{t}$
until it saturates on a timescale of order $L^2$.  
In contrast, in phase III, any deviations from the NEAS can produce operators that
spread ballistically to the boundary, so initially the increase in entropy can grow linearly in $t$.  This fast entropy production will saturate
on a timescale of order $L/v_B$ as the system is brought to a local equilibrium.  But if the magnetization profile in this initial local equilibrium is extensively away from the steady state, then at
later times, diffusive spin transport is required for further entropy increase, leading to a crossover in time to diffusive relaxation of the local equilibrium state, which then finally approaches a NEAS by a timescale of order $L^2$.

\section{Extensions of the random circuit model}
\label{sec:ext}

Two natural extensions of the random circuit model are to consider  higher-dimensional lattices, with the left and right boundaries of the lattice coupled to reservoirs, or to allow each site to also be coupled to a qudit with local Hilbert space dimension $q$, which acts as a local bath for the qubit.  The closed system dynamics for the qudit model was studied in  Ref.~\cite{Khemani17} by replacing $U_+$ and $U_-$ in the definition of the two-qubit unitaries in the bulk [see  Eq.~(\ref{eqn:U})] by two independent Haar random unitaries acting on the qudits.  The central 2$\times 2$ matrix was replaced by a Haar random unitary acting on the $2q^2$-dimensional space space spanned by $\lvert{\uparrow\downarrow}\rangle \otimes\ket{n n'}$ and $\lvert{\downarrow \uparrow}\rangle \otimes \ket{n n'}$, where $n$ and $n'$ represent the state of the qudit on the two sites.  

\begin{figure}[tb]
\begin{center}
\includegraphics[width = 0.25 \textwidth]{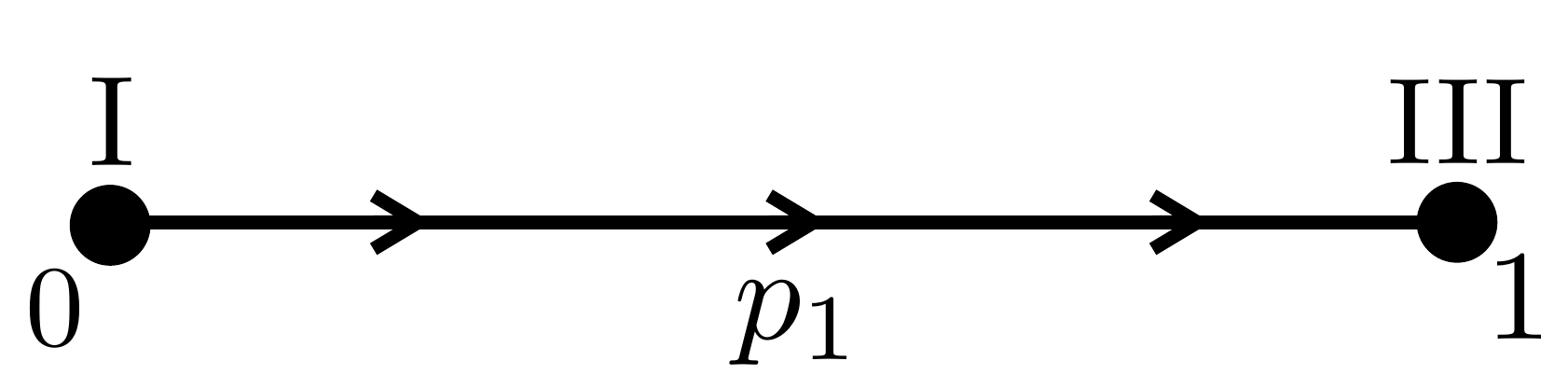}
\caption{Expected NEAS phase diagram as a function of the partial swap probability $p_1$ for the random circuit model in higher spatial dimensions or with an additional qudit at each site in arbitrary dimensions \cite{Khemani17}, which acts like a local bath for the qubit.  Although phase I is preserved for $p_1=0$, there is no noninteracting fermion regime (phase II) in these models, and we expect that the system immediately becomes quantum chaotic and achieves local equilibrium (phase III) for any finite $p_1$.}
\label{fig:phaseExt}
\end{center}
\end{figure}

In the qudit model in arbitrary spatial dimensions, phase I still naturally appears by forbidding partial swaps of the state of the qubit.  However, adding any finite probability $p_1$ for partial swaps  results in quantum chaotic dynamics even for $p_2=0$, which drives the system to local equilibrium (phase III, see Fig.~\ref{fig:phaseExt}).  Effectively, the qudits act as a quantum chaotic bath, which can rapidly dissipate entropy into the reservoirs through ballistic operator spreading.  
For higher-dimensional lattices of qubits with finite $p_2$, phase I and phase III are naturally realized for $p_1=0$ and any nonzero $p_1$, respectively.   Furthermore, for $p_2=0$ beyond 1D, there is no longer a mapping of the random circuit to a system of noninteracting fermions, and the operator spreading becomes ballistic for any finite value of $p_1$.  In this case, we expect that the system directly realizes phase III even along the $p_2=0$ axis.  The entropy and mutual information properties of the qudit models and the higher-dimensional qubit models can  be studied analytically using  extensions of the $abc$ model discussed in Appendix \ref{app:phaseIII}.  Although we do not expect the higher-dimensional random circuit models to naturally realize phase II, one can consider explicit noninteracting fermion models in higher dimensions to realize this phase, one example of which is discussed in Sec.~\ref{sec:anderson}.

\section{noninteracting Anderson model~}
\label{sec:anderson}

The presence of the volume-law mutual information in phase II  arises from the diffusive spreading of the fermions that were produced by a Jordan-Wigner transformation of the spins.   It is then natural to ask whether a similar effect occurs in a paradigmatic model for diffusive fermions: the Anderson model in 3D, with the Hamiltonian
\be
H= t_0 \sum_{\mean{\bx \by}}c_\bx^\dagger c_\by + \sum_\bx V_\bx c_\bx^\dagger c_\bx ,
\ee
where $t_0$ is the hopping, the first sum is over nearest-neighbor sites of a simple cubic lattice, and $V_\bx$ is a random potential on each site.  We  draw $V_\bx$ from a uniform distribution between $\pm W/2$.  
In 3D, this model has a metal-insulator transition near $W_c/t \approx 16.5$ \cite{MacKinnon81,Pichard81,Slevin14}.  We focus on the diffusive regime $0<W< W_c$.  Similar to the averaging over random unitaries in the circuit model, this model has characteristic disorder-averaged properties. For example, the disorder-averaged density-density response function in 3D at small $(k, \omega)$ is  dominated by a diffusive pole
\be
\overline{\chi(k,\omega)} \sim  \frac{4\pi \nu}{D k^2 - i \omega},
\ee
where, in this section, the overbar denotes a disorder average, $D$ is the diffusion constant, and $\nu$ is the density of states, with the fluctuations of these diffusive modes well described by a nonlinear sigma model \cite{EfetovBook}. 

\subsection{Entropy from scattering states}

To determine the entanglement structure of the NESS for this Anderson model in contact with two reservoirs, we use a description of the eigenstates of the system in terms of scattering states  \cite{Buttiker92,Beenakker97,Bruch18}, as shown in Fig.~\ref{fig:scattering}.  The leads are taken to be ballistic conductors with the same Hamiltonian as the disordered region but with no random potential: $V_i=0$.  The incoming scattering states at a given energy $E$ in transverse channel $n$ are defined by a set of fermion operators $a_\alpha^n(E)$ whose wave function satisfies the boundary condition that it is an incoming plane wave in lead $\alpha$ and channel $n$.  These operators have the correlation functions
\begin{align}
\mean{a_\alpha^{n \dag}(E) a_\beta^m(E')} &= \delta(E-E') \delta_{\alpha \beta} \delta_{m n} \sigma_\alpha(E),\\
\sigma_\alpha(E)& = \frac{1}{e^{(E-\mu_\alpha)/T}+1},
\end{align}
where $\mu_\alpha$ is the chemical potential of lead $\alpha$, and $T$ is the temperature, which, for simplicity, we take to be the same in the two leads.
The outgoing operators $b_{\alpha}^n(E)$ are related to these incoming  operators through the $S$ matrix.  

\begin{figure}[b]
\begin{center}
\includegraphics[width = 0.48 \textwidth]{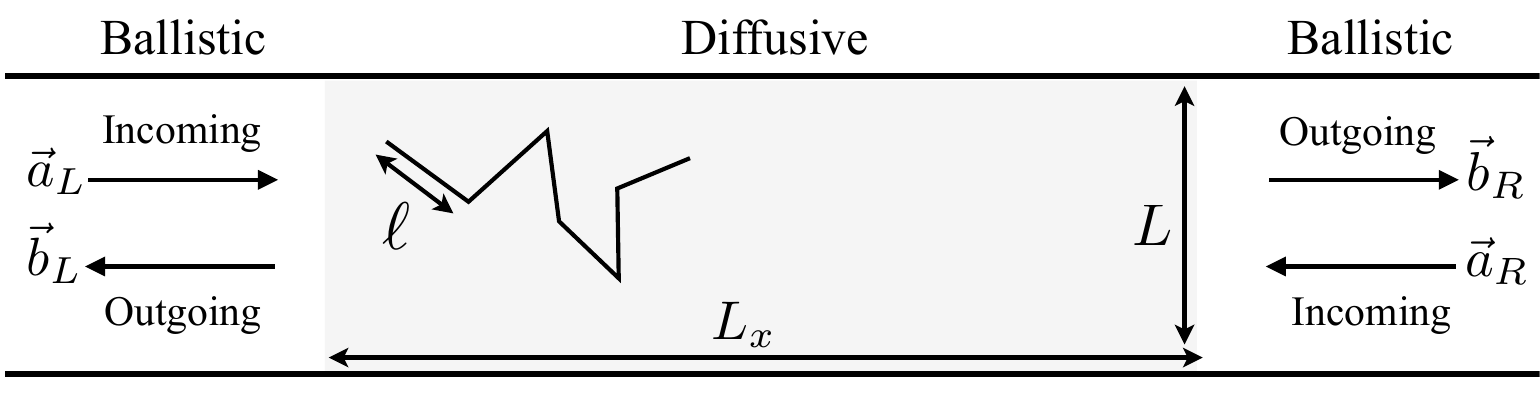}
\caption{Incoming and outgoing scattering states for a disordered wire with mean free path $\ell$. Here, $L_x$ is the longitudinal length of the wire and $L$ is the transverse width.}
\label{fig:scattering}
\end{center}
\end{figure}

As in phase II, the reduced density matrix for the disordered region is entirely determined by the two-point function.  Defining the wave function of the incoming scattering states
\be
\phi_\alpha^n(\br,E) \equiv \bra{\br} a_\alpha^{n \dag}(E) \ket{0},
\ee
the two-point function inside the wire is
\begin{align} \nonumber
\mean{c^{\dag}_{\bx}  c_{\by}}& =\sum_n  \int  dE  \big[\delta \sigma(E)  \phi_L^{n *}(\bx,E) \phi_L^n (\by,E) \\
&+ \sigma_R(E) \sum_\alpha \phi_\alpha^{n *}(\bx,E) \phi_\alpha^n (\by,E)\big],
\end{align}
where $\delta \sigma = \sigma_L(E)-\sigma_R(E)$ is the difference in the Fermi functions of the two leads and the scattering states are normalized to have unit current in the leads \cite{Pendry92}. 
 The term proportional to $\delta \sigma$ gives the nonequilibrium contribution to the two-point function.  
 The excess fermion density in the wire due to the bias (assuming $\delta \mu \equiv \mu_L-\mu_R>0$, and taking as our equilibrium reference the state when $\mu=\mu_R$ in both leads) is given by
\be \label{eqn:dnx}
\delta n(\bx) =\sum_n \int  dE \delta \sigma(E) \abs{\phi_L^n(\bx,E)}^2.
\ee
After disorder averaging, $\overline{\delta n(\bx)}$ will have the linear ramp profile discussed in Sec.~\ref{sec:summary}.   

Similar to phase II, we can compute the corrections to the entropy by expanding the two-point function around its disorder average.  The correction to the entropy of a given subregion $A$ will then be approximately equal to [see Eq.~(\ref{eqn:entNIF})]
\be
\Delta S_A \approx 2 \sum_{\bx,\by \in A} \overline{\lvert{\mean{c_{\bx}^{\dag} c_{\by}}}\lvert^2} - \overline{n(\bx)}^2 \delta_{\bx \by},
\ee
 which depends on the disordered average wave-function correlations between different channels and energies.  We parametrize these correlations by introducing the function 
 \be
\begin{split}
\Phi{}_A^{n m}(E,\Delta )&= \sum_{\bx,\by \in A}   \overline{\phi_L^{n *}(\bx,E) \phi_L^{m }(\bx,E+\Delta ) }  \\
&\overline{\times \phi_L^n (\by,E) \phi_L^{m *} (\by,E+\Delta )}\\
&- \sum_{\bx \in A}  \overline{\abs{\phi_L^n(\bx,E)}^2}~  \overline{\abs{\phi_L^{m}(\bx,E+\Delta )}^2} .
\end{split}
\ee
In the limit where $\delta \mu,T \ll t_0 \sim W$, over the range of energies that contribute to the NESS, $\Phi_{A}^{nm}(E,\Delta )$ will only depend on the energy difference $\Delta $.  Taking the weakly driven limit $\delta \mu \ll T \ll t_0 \sim W$ to compare to the random circuit, this allows us to approximate the entropy deviation from the disorder-averaged reduced-density matrix for region $A$ as
 \be
 \Delta S_A \approx  \frac{2 \delta \mu^2}{T}\sum_{nm} \int d \Delta \,   \Phi_A^{nm}(\mu_R,\Delta ).
\ee
The scaling of $\Delta S_A$ with the geometry of $A$ is determined by several factors.   First, from the definition of the scattering states, the wave-function amplitude scales as $\phi_L^n(\bx,E )\sim 1/L$, where we have taken the two transverse linear dimensions of the leads and the diffusive region to be $L$.   Assuming $A$ is a finite fraction of the diffusive region, which is of length $L_x$, the sum over $\bx$ and $\by$ involves on the order of $L_x^2 L^4$ terms, which  implies $\Phi_A^{nm}(\mu_R,0) \sim L_x^2 $.  As we argued in Sec.~\ref{sec:summary}, the energy range over which the scattering-state wave functions change on the length scale $L_x$ is given by the Thouless energy $E_{\rm Th} \sim 1/L_x^2$, which allows us to represent $\Delta S$ as 
\begin{align}
\Delta S_A &\approx \frac{2   t_0 \delta \mu^2}{T L_x^2}  \sum_{nm}  \Phi_A^{nm}(\mu_R,0) \int d\epsilon\, s_A(\mu_R,\epsilon) , \\
s_A&(\mu_R,\epsilon) = \frac{\sum_{nm}\Phi_A^{nm}(\mu_R,\epsilon t L_x^{-2})}{\sum_{nm} \Phi_A^{nm}(\mu_R,0)},
\end{align}
where $s(\mu_R, \epsilon)$ is independent of $L_x$ in the scaling limit, but will depend on the aspect ratio $L/L_x$ of the diffusive region.  
For diffusive wires in the regime $L \ll L_x$, the transmission through the wire satisfies the ``isotropy'' condition \cite{Mello91}, which states  that the wave functions in different transverse channels are completely uncorrelated at each transverse slice of the wire.  This assumption permits analytic treatments of the transmission and conductance through the wire using random matrix theory (RMT) \cite{Beenakker97}.  For the scaling analysis of the entropy, this would imply that the sum over $n$ and $m$ can be reduced to a single sum over $n$, leading to an area-law scaling with $\Delta S_A$ independent of $L_x$ in this $L_x \gg L$ regime.  However, in the regime $L\sim L_x$, this isotropy condition is known to break down \cite{Jalabert95}, which allows additional correlations that violate the area-law scaling.

Using numerical solutions of the scattering-state wave functions based on the transfer matrix method \cite{MacKinnon81,Pichard81,Pendry92,Slevin14}, we have confirmed that there is indeed a volume-law scaling of the mutual information of two halves of the wire when the diffusive region has an aspect ratio of 1.  Our numerical results are summarized in Fig.~\ref{fig:anderson}.  We consider a wire on a cubic lattice of size $(L+2)\times L \times L$ with periodic boundary conditions in the transverse direction and strong disorder $W =8 t_0 $, but still well within the metallic phase ($W_c \approx 16.5\, t_0$).   We  parametrize the average mutual information between two regions as
\begin{align} \label{eqn:iabd}
\overline{I(A:B)} &\approx   \frac{  \nu^2 t_0 \delta \mu^2}{T  }    f_{AB}(\mu_R) L^2 \int d \epsilon\, i_{AB}(\mu_R,\epsilon), \\
i_{AB}(\mu_R,\epsilon) & = \frac{\sum_{nm} \Delta \Phi_{AB}^{nm}(\mu_R,\epsilon t L_x^{-2})  }{\sum_{nm} \Delta \Phi_{AB}^{nm}(\mu_R,0) }, \\
\Delta \Phi_{AB}^{nm} &=\Phi_{AB}^{nm} - \Phi_{A}^{nm} - \Phi_B^{nm}, \\
f_{AB}(\mu_R) &= \frac{2}{L_x^2 L^2} \sum_{nm} \Delta \Phi^{nm}_{AB}(\mu_R,0) ,
\end{align}
where $i_{AB}(\mu_R,\epsilon)$ converges to a single function of $\epsilon$ that is independent of $L$ in the large-$L$ limit.

\begin{figure}[bt]
\begin{center}
\includegraphics[width = .4 \textwidth]{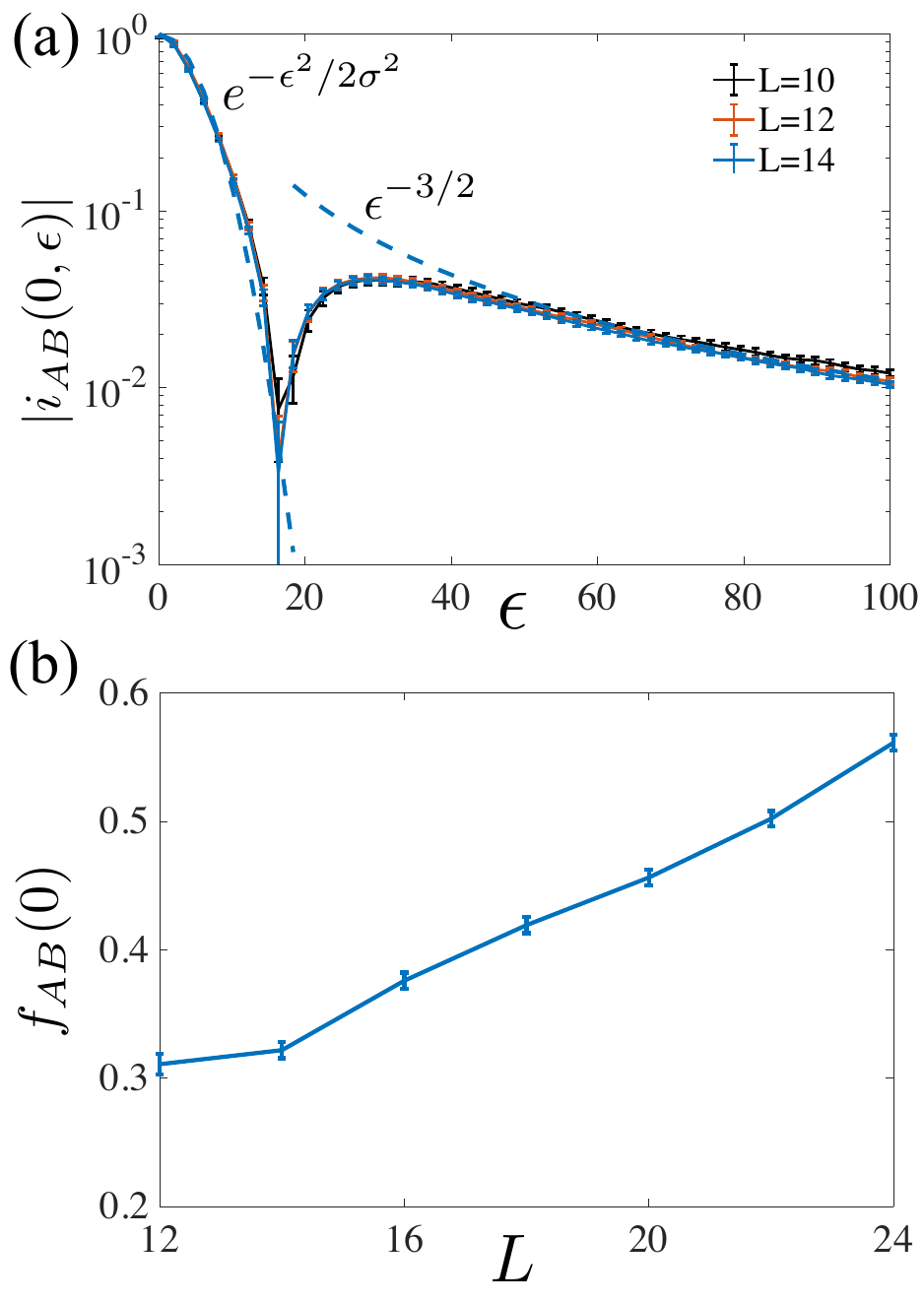}
\caption{(a) Energy correlation function $\abs{i_{AB}(0,\epsilon)}$ for the mutual information between two halves of the central $L\times L \times L$ region of an $(L+2)\times L \times L$ wire with $t_0=1$ and $W=8 $ for 400 realizations of disorder. The approximate zero in the absolute value is due to a sign change.  (b) Finite-size scaling of $f_{AB}$, which determines the overall scaling of the mutual information.  Parameters as in (a) with 1000 disorder realizations for each value of $L$. Error bars denote $\pm$3 standard deviations. }
\label{fig:anderson}
\end{center}
\end{figure}

Taking $A/B$ equal to the left/right half of the central $L\times L \times L$ cube of the wire, Fig.~\ref{fig:anderson}(a) shows the scaling of $\abs{i_{AB}(\mu_R,\epsilon)}$ with $L$, demonstrating that it quickly converges to a universal curve.   The converged function for $i_{AB}(0,\epsilon)$ is well fit by a Gaussian decay $e^{-\epsilon^2/2\sigma^2}$ at small values of $\epsilon$ with $\sigma \approx 5$,  but then changes sign and crosses over to a power-law tail  that is well fit by $\epsilon^{-3/2}$ over this range of $\epsilon$.  We note that a similar crossover at the Thouless energy to an $\epsilon^{-d/2}$ tail  has been observed in the  spectral statistics of small metallic particles in dimension $d$ \cite{Jalabert93}.     For the parameters used in Fig.~(\ref{fig:anderson}), we find that the integral of $i_{AB}(0,\epsilon)$ over $\epsilon$ is approximately equal to 4.
The overall scaling of the mutual information with $L$ is set by the term $f_{AB}$, which is shown in Fig~\ref{fig:anderson}(b).  Fitting to $f_{AB} = a L + b$, we determine the coefficients  
\be
a \approx 0.021, ~b \approx 0.038,
\ee
consistent with an overall volume-law scaling for the mutual information.  We leave a direct calculation of the logarithmic negativity $\E(\rho)$ for future work.   However, because of the similar structure of the fermionic two-point function of this model with phase II in the random circuit, we expect that the NESS will similarly exhibit a volume-law scaling for entanglement.

\subsection{Scattering-state correlations in the random circuit}

In the case of the 3D Anderson model, we saw that the volume-law mutual information arose from ``hidden'' correlations between different transverse channels of the scattering states.  Here, we show that a similar effect occurs in phase II of the random circuit model.

For the random circuit  we do not have an identical notion of a scattering state as in extended Hamiltonian systems.  Instead, we define our scattering states to be the set of all reservoir operators at $t=0$.  Running the circuit forward in time, these reservoir operators will diffuse into and eventually out of the system on a timescale of about  $L^2$.  At any instant of the circuit, we thus have an over-complete set of operators given by the of order $L^2$ reservoir operators that are still present and diffusing in the system.  We can roughly understand the evolution of the probability amplitude of a single reservoir operator by coarse graining its evolution to a diffusion equation.  In this case, it satisfies a standard diffusion equation with absorbing boundaries.  Using the method of images, this diffusion problem has the solution at early times 
\be \label{eqn:res}
|\phi(x,\tau)|^2 \sim \frac{x}{\tau^{3/2} L^2} e^{-x^2/4 {D} \tau} ,~\tau \ll 1
\ee
where $\tau=n/L^2$ and $x=i/L$ are  rescaled time $n$  and space $i$ variables, respectively.  At late times, the dynamics are dominated by the slowest diffusive mode that satisfies the boundary condition that it vanishes at the edges
\be
|\phi(x,\tau)|^2 \sim \frac{\sin \pi x}{L^2} \, e^{- \pi^2 {D} \tau},~  \tau \gg 1.
\ee
Thus, we see that the scattering states spread across the system on the diffusive timescale and have an amplitude of about $1/L$ before eventually leaking back out to the reservoirs. 

Similar to Hamiltonian systems, we define an incoming scattering-state wave function using the backward time evolution of the system operators
\be
c_i^n = \sum_\ell \phi_\ell^i(n-m) c_\ell^{m},
\ee
where $c_\ell^m$ are the operators at time step $m<n$ including both the system and the reservoir.  Evolving sufficiently far backward in time that $\phi_\ell^i(n-m)$ only has sizable overlap with the reservoir, we can express the two-point function in the system as
\be \label{eqn:cij}
\mean{c_i^{n \dag} c_j^n} =  \sum_{\ell \in {\rm Res}} \phi_\ell^{i *}(n-m) \phi_\ell^{j}(n-m)  \mean{c_\ell^{m\dag}c_\ell^m},
  \ee
  where, similar to the Anderson model, we used the fact that scattering states in the reservoir are in an uncorrelated product state at negative times.
   For $i=j$, we see that Eq.~(\ref{eqn:cij}) is a sum of order $L^2$ terms with an amplitude of about $1/L^2$, and thus we find that the density in the system is always order one, as expected. 
 For $i\ne j$, we can use the identity
  \be
  \{ c_i^{n \dag},c_j^n \} = 0 = \sum_{\ell \in {\rm Res}} \phi_\ell^{i *}(n-m) \phi_\ell^{j}(n-m),
    \ee
 to  rewrite the coherences for anti-symmetric reservoirs as 
  \be
  \begin{split}
  \mean{c_i^{n \dag} c_j^n} =  \delta  \sum_{\ell \in {\rm Left}}  \phi_\ell^{i *}(n-m) \phi_\ell^{j}(n-m) .
  \end{split}
  \ee
  We immediately see that at equilibrium $(\delta =0)$, the coherences vanish as expected.  
For finite $\delta$, we naively have a sum over $L^2$ terms with random phases and amplitude $1/L^2$, which would result in the scaling
  \be
\lvert{  \mean{c_i^{n \dag} c_j^n}}\lvert^2 \sim \delta^2/L^2.
\ee
 Such a scaling would lead to an area law for the mutual information, in sharp contrast to what we found in Sec.~\ref{sec:phaseII} in Eq.~(\ref{eqn:axy}).   The resolution to this paradox is similar to the effect we found for the 3D Anderson model, where the volume law  arose due to the presence of correlations between transverse channels.  In the case of the random circuit, the volume-law mutual information is recovered because of correlations between reservoir wave functions that persist over a finite time.

\section{Approach to local equilibrium following a quench}
\label{sec:appLE}

In this section, we go beyond the scenario of the boundary-driven system considered in most of this work to qualitatively consider the implication of our results for the time dynamics of interacting diffusive systems when the initial density profile has deviations from equilibrium.  For weak interactions, either in the random circuit or in disordered metals, the butterfly velocity is parametrically smaller than other characteristic scales for the diffusive dynamics.  For the random circuit, the diffusion constant is set by the lattice spacing and the rate at which the unitaries act, while the butterfly velocity is reduced by a factor of $\sqrt{p_2}$ according to Eq.~(\ref{eqn:vb}).  In a weakly interacting diffusive metal the diffusion constant is determined by the elastic scattering rate and the mean free path, while $v_B \sim \sqrt{D \gamma_{\rm in}}$, where $\gamma_{\rm in}$ is the inelastic scattering rate due to interactions \cite{Swingle17,Patel17}; here, we consider the case where the elastic scattering rate is large compared to $\gamma_{\rm in}$.  Using the mapping of the local random circuit to  diffusive fermions, we can interpret the parameter $p_2$ as the analog of $\gamma_{\rm in}$.   As we now argue, this parametric separation between operator spreading dynamics and the diffusive dynamics leads to a large window of space and time where the system will display strong deviations from local equilibrium with a similar structure to the volume-law mutual information phase of the NEASs.

As a thought experiment we imagine a quench experiment in the random circuit model whereby
we prepare the left and right halves of the chain in initial mixed-product-state diagonal density matrices with different magnetizations $m_{L/R} = (1\pm \delta)/2$ on the left/right half.  
At time scales much larger than the Thouless time, the magnetization will relax to either equilibrium if the system is not driven or to the NEAS if it is driven; however, in either case, the early-time dynamics on the diffusive length scale will display behavior reminiscent of the boundary-driven system.  Letting the step in the magnetization occur at $x=0$, the diffusion length $\ell_D \equiv \sqrt{D t}$ then sets the length scale for the diffusive smoothing of the magnetization step. 
For weak interactions the diffusion length is large compared to $v_B t$ at early times, and crossover length and timescales, 
\be
\ell_{\rm in} \equiv D/v_B,~\tau_{\rm in} \equiv D/v_B^2~,
\ee
are set by $\ell_D(t)=v_Bt$.
For $t \ll \tau_{\rm in}$, we can neglect the effects of the ballistic operator spreading in describing features on the scale $\ell_D$, thus, the quench dynamics will behave similarly to the volume-law mutual information phase (phase II) of a boundary-driven system with $L \sim \ell_D$.  On the other hand, for $t \gg \tau_{\rm in}$ the operator spreading will have scrambled most of the information originally encoded on the length scale $\ell_D$ over a much larger range on the order of $v_B t$.  Thus, the system will appear to be in local equilibrium, similar to phase III of the boundary-driven system.  
Note that in the random circuit model the elastic mean free path is on the order of a single lattice spacing, which implies that diffusion is well defined down to this scale.    This case should be contrasted with the strongly interacting regime for diffusive metals, where the interactions can lead to a renormalization of the diffusion constant.  In that context, $\tau_{\rm in}$ can instead be identified with a microscopic equilibration time $\tau_{\rm eq}$, which then provides a constraint on the renormalization group flow by  implying the existence of an upper bound on the diffusion constant $D \lesssim v_B^2 \tau_{\rm eq}$ \cite{Hartman17}.  

\begin{figure}[tb]
\begin{center}
\includegraphics[width = 0.49 \textwidth]{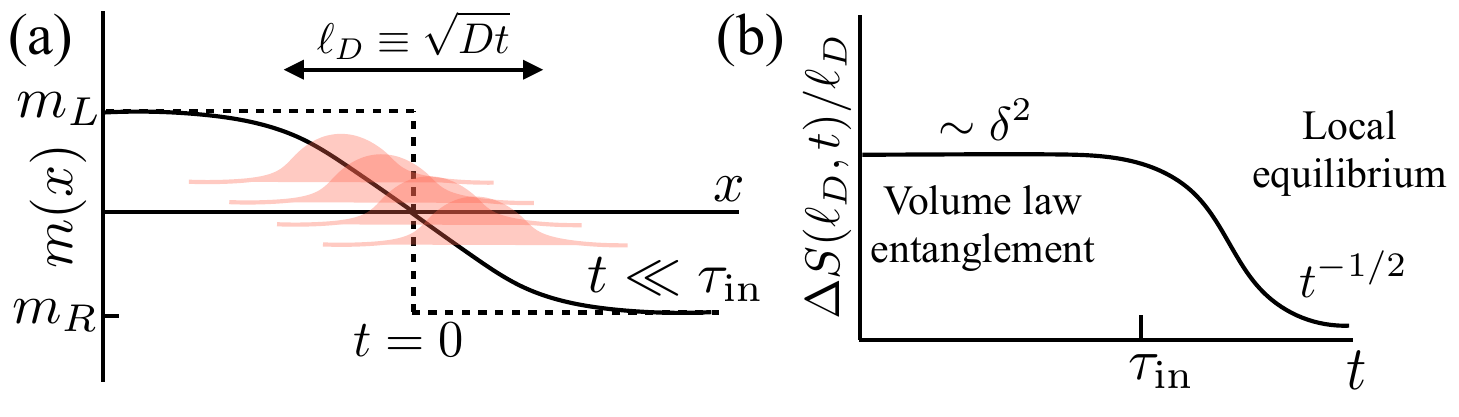}
\caption{(a) Quench dynamics for a diffusive system with weak interactions and a single conserved quantity, initially with a step profile in the conserved quantity.  The step broadens diffusively and, on short timescales compared to the inelastic scattering time $\tau_{\rm in}=D/v_B^2$, the system will exhibit the long-range coherence and entanglement characteristic of phase II of the boundary-driven system.  (b) Deviation of the entropy from local equilibrium $\Delta S(\ell_D,t)$ at the diffusion length, showing the crossover from volume to area law near $t \sim\tau_{\rm in}$.   }
\label{fig:quench}
\end{center}
\end{figure}

As shown in  Fig.~\ref{fig:quench}(a), the initial step in density will broaden diffusively over a length $\ell_D$ so that the instantaneous current density  near the origin at time $t$ will scale as $J \sim \delta/\ell_D$. 
Defining the reduced density matrix at time $t$ over the region between $\pm x$ as $\rho(x,t)$,
 the entropy deviation from local equilibrium is
\be
\Delta S(x,t) \equiv S[\rho_{\rm LE}(x,t)] - S[\rho(x,t)].
\ee
 From our analysis of the boundary-driven system, we then have the result that for times much less than $\tau_{\rm in}$, the entropy deviation for $x = \ell_D$ will scale as
\be
\Delta S(\ell_D,t) \sim \delta^2 \ell_D,~t \ll \tau_{\rm in}.
\ee
Thus, we expect the system exhibits volume-law deviations from local equilibrium on this length scale.  
On the other hand, for times much longer than $\tau_{\rm in}$ but short compared to the Thouless time, the situation is analogous to phase III of the boundary-driven system, so we  only have area-law deviations from local equilibrium, Fig.~\ref{fig:quench}(b).

 From this picture, we can also develop some intuitive understanding for the finite-size scaling that sets the crossover with increasing $p_2$ from phase II to phase III for the NEAS of the boundary-driven system.  In this case, the Thouless time $\tau_{\rm Th} = L^2/D$ imposes a cutoff on the time dynamics, which suggests that the crossover occurs precisely when $\tau_{\rm Th} \sim \tau_{\rm in}$.  This observation implies that the crossover function for the NEAS density matrices starting from phase II will be a universal function of $p_2 L^2$ in the scaling limit.  This scaling is directly verified in Sec.~\ref{sec:crossovers}.

\section{Experimental signatures}
\label{sec:exp}

In this section, we outline an experimental approach,  suitable for mesoscopic wires or ultracold Fermi gases \cite{Gross17},  to probe two-site coherences in the NEAS.  When the sites are separated by a distance much greater than $\ell$, this measurement is an indirect probe of the long-range entanglement in the system.    

The idealized setup is shown in Fig.~\ref{fig:exp} and consists of four essential components: (i) a disordered 3D metallic system being driven by a chemical potential bias (here we focus on energy-conserving systems at a fixed temperature), (ii) two 1D channels connected to the wire via tunable, tunneling barriers,  (iii) a coherent beam splitter \cite{Oliver99,Henny99b,Cassettari00,Houde00}, and (iv) a pair of current or charge sensors.  
Referring to the length scales $(\ell,\ell_{\varphi},\ell_{\rm ee}, \ell_{\rm ep})$ introduced in the Introduction, our analysis of the random circuit model implies that achieving volume-law entanglement in the NEAS requires a wire of length $\ell \ll L \ll \ell_{\varphi}$.  In typical metals, these length scales are (50 nm,~1~$\mu$m,~10~$\mu$m,~10 mm) \cite{Altshuler82,Steinbach96}; thus, the regime of interest is easily achievable and, in fact, has been directly probed in an experiment conceptually similar to the one we propose here \cite{Pothier97}.  In atomic Fermi gases the strength of interactions are tunable via a Feshbach resonance \cite{Chin10}, which would allow one to scan $\ell_{\varphi}$ for a fixed length of the channel.

\begin{figure}[tb]
\begin{center}
\includegraphics[width = 0.3 \textwidth]{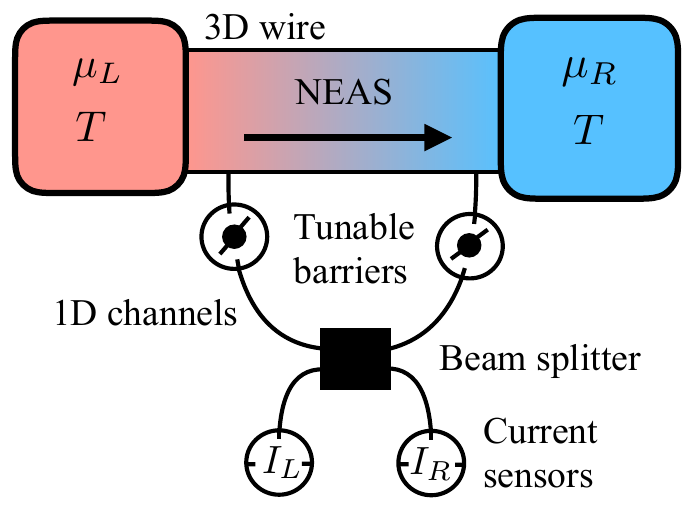}
\caption{Idealized experiment to probe the long-range coherences of a current-driven disordered fermionic system at fixed temperature $T$ with a chemical potential bias.  The channel is taken to either be formed from a 3D metallic wire or a 3D disordered potential landscape in an atomic Fermi gas.   After the system reaches the NEAS, tunable tunneling barriers into 1D channels are opened to allow an interferometric measurement of the long-range coherence in the wire for a specific realization of the disorder. 
}
\label{fig:exp}
\end{center}
\end{figure}

For simplicity, we focus on a measurement approach where the tunneling to the two 1D channels is turned off while the system is driven to its steady state by the bias.  After a sufficiently long time, the tunneling barriers are lowered, allowing current leakage into the two 1D channels. To avoid interaction effects in the channel, we consider the limit of weak tunneling into the channels as compared to the transit time through the beam-splitter device.
We assume the left and right 1D channels to be connected to single sites $L$ and $R$ of the wire, respectively. Taking a beam splitter with reflection coefficient $r$ and transmission coefficient $t$, the current in the left or right meter will be proportional to
\begin{align}
I_L &\propto |r|^2 \mean{\psi_L^\dag \psi_L} +  |t|^2  \mean{\psi_R^\dag \psi_R} + 2 \mathrm{Re}[ r^* t \mean{\psi_L^\dag \psi_R}] , \\
I_R &\propto |t|^2 \mean{\psi_L^\dag \psi_L} + |r|^2  \mean{\psi_R^\dag \psi_R} +2 \mathrm{Re}[ t^* r \mean{\psi_L^\dag \psi_R}].
\end{align}
For a balanced beam splitter (i.e., $|r|^2=|t|^2=1/2$) we then have the result that 
\be
I_L - I_R \propto 2 \mathrm{Re}[(r^*t -  t^* r) \mean{\psi_L^\dag \psi_R}].
\ee
Thus, by tuning the transmission and reflection phases of the beam splitter (e.g., by varying the path length or tunneling barriers), it becomes possible to directly measure the coherence between two sites of the wire for a static disorder potential.
 
The setup discussed here  demonstrates that these long-range coherences are physically accessible in a transport experiment.  
A more detailed analysis of this measurement scheme is beyond the scope of this work.  In general, these effects should be most easily observed in diffusive wires with a short elastic mean free path.  In addition, one can consider geometries that use multiple channels and beam splitters to enhance the total detected current difference.  An alternative approach is to consider driven, time-dependent systems in a 1D or quasi-1D system similar to the random circuit model.

\section{Conclusion}
\label{sec:outlook}

While systems in thermodynamic-equilibrium Gibbs states  have extensive von Neumann entropy, they usually satisfy area laws for their entanglement entropy or mutual information.  On the other hand, any modification that drives the system out of equilibrium allows for potential violations of this behavior.  In this paper we investigated a common nonequilibrium scenario consisting of an extended system coupled at its two ends to reservoirs at different chemical potentials,  leading to current-carrying nonequilibrium attracting states (NEASs) in the long-time limit.  The analog of thermalization for these current-driven systems is the approach to local equilibrium.
For a family of random circuit models, we found that the ballistic operator spreading associated with quantum chaos is crucial for the emergence of local equilibrium in these models.  As our argument in favor of local equilibrium is rather general, it will be interesting to test these ideas in quantum chaotic Hamiltonian or Floquet driven systems.  In addition, given that we found the mutual information of the NEASs in the quantum chaotic region is consistent with an area law, it is likely that these states can be well approximated using a matrix product representation, allowing a rich set of analytical and numerical techniques to be applied to further characterize these states. 

A more surprising result is that, in one of the nonchaotic phases (phase II), we found that the system is driven towards a highly entangled state with a volume-law mutual information and logarithmic negativity between two halves of the chain.  
We showed that this phase captures important qualitative features of the entanglement structure of the current-driven 3D Anderson model.
As a result, we expect signatures of this phase to appear in a wide variety of metallic systems on mesoscopic length scales and timescales.  We discussed  an experimental approach applicable to mesoscopic wires or ultracold Fermi gases where one could directly probe these effects.
More generally, this behavior demonstrates the ability to stabilize a high degree of entanglement in nonequilibrium systems with limited fine-tuning of the evolution.   Determining other conditions under which one can achieve such a strong violation of local equilibrium is a promising direction for future research.  
Finally, one is naturally led to ask whether  such  large deviations of the entropy away from local equilibrium or  extensive entanglement   can be harnessed as a thermodynamic resource or for applications in quantum information science.  For example, we showed in Sec.~\ref{sec:exp} that the long-range coherences in the NEAS could be transformed into a current difference between two 1D channels, which can be directly used to perform work.

Future theoretical work on the Anderson model could investigate the behavior of the NESS entanglement across the metal-insulator transition, where, naively, one expects a crossover to an area-law scaling due to the localization of the wave functions.  Using the mapping of disorder-averaged Anderson models to supersymmetric field theories \cite{EfetovBook}, it will be interesting to understand how the volume-law mutual information of the NESS emerges in the corresponding field theoretic description.    

Another interesting avenue of investigation for the random circuit models is to better understand their connection to classical boundary-driven stochastic lattice gases. In particular, it may be fruitful to investigate a boundary-driven random circuit model with a true phase transition.  One prominent example in the classical case is the asymmetric simple exclusion process (ASEP) in which the particles have unequal probabilities of hopping left or right \cite{Derrida07}.  Naively such a model would need to break unitarity in the random circuit, however, it may be possible to engineer chiral transport using ladder systems.  A related question to explore is the role of dissipation in the bulk of the system, which can be easily incorporated in the random circuit model  by allowing dissipative quantum channels.  We gave one such example in Sec.~\ref{sec:crossovers} as an effective model for the crossover from the discrete hopping limit of the random circuit (phase I) to the quantum chaotic phase (phase III).

In the context of open quantum systems, our work introduces the concept of NEAS density matrices to the description of the long-time behavior of nonequilibrium open systems with noise or time dependence in their parameters.  In time-independent systems, there is a natural classification of NESSs into those described by low-entropy mixtures of a few pure states and high-entropy mixed states, with the latter often having an effective thermal description at long wavelengths \cite{Sieberer16}.  Extending this classification program to the description of NEASs may prove a rich direction of research.  

\emph{Note added:} During the final stages of publication of this paper, two related papers appeared \cite{Gullans19,Bernard19}.  In Ref.~\cite{Gullans19}, we extend the analysis of boundary-driven 3D Anderson models to consider the entanglement scaling across the localization transition.  Ref.~\cite{Bernard19} introduces another class of stochastic, noninteracting fermion models whose average dynamics reduce to the SSEP and whose NEASs have similar properties to the noninteracting fermion phase (phase II) of the random circuit model studied in this work.

\begin{acknowledgements}
We would like to thank Vedika Khemani, Joel Lebowitz, Adam Nahum, Toma{\v z} Prosen, Brian Swingle, Mike Zaletel, and Marko {\v Z}nidari{\v c} for helpful discussions and correspondence.  MJG thanks Raghu Mahajan for helpful discussions related to Sec.~\ref{sec:appLE}. We thank Arnold Mong for preliminary numerical work on a related Floquet system.  Research was supported by National Science Foundation Grant No.~DMR-1420541.
\end{acknowledgements}

\appendix

\section{Overview of Appendixes}
\label{app:overview}

The appendixes are organized as follows: In Appendix \ref{app:thms}, we present several general theorems that prove the existence of a set of NEAS density matrices for the local random circuit model.  In Appendix \ref{app:SSEP}, we review some basic properties of the SSEP, which describes the time-averaged behavior of the random circuit model.  In Appendix \ref{app:NIFder}, we provide more details for the analysis of the scaling limit in phase II.  In Appendix \ref{app:phaseIII}, we derive a compact representation of the replicated density matrix that describes both the quantum chaotic phase III and the crossovers between the three phases. In Appendix \ref{app:open}, we present an analysis of the von Neumann entropy of an open random circuit without charge conservation.  For a zero entropy reservoir, we find that the entropy of the system is reduced by one bit, on average, compared to the infinite-temperature state.

\section{Stationary  random quantum circuits}
\label{app:thms}

In this appendix, we establish some basic facts about the long-time behavior of random quantum circuits coupled to reservoirs.  
This analysis demonstrate that the models considered in this paper  have an attractive ensemble of density matrices in the long-time limit, which we refer to as nonequilibrium attracting states (NEASs).

For a given $d$-dimensional quantum system, an associated quantum channel $\E(\cdot)$ is defined as a trace-preserving, completely positive linear map that acts on the set of $d \times d$ complex matrices $M_d(\C)$.  Every quantum channel has a representation in terms of a collection of Kraus operators $\E_k$ satisfying $\sum_r K_r^\dag K_r =\id$ such that the density matrix of the system is transformed as \cite{Wolf12}
\be
\E(\rho) = \sum_r K_r \rho K_r^\dag.
\ee
In this work, we are interested in characterizing fixed points $\E(\cdot)$, defined as
\be
\mathcal{F}_{ss}(\E)=\{ {\rm Density~matrices}~\rho_{ss}: \E(\rho_{ss})=\rho_{ss} \}.
\ee
Since the set of density matrices is a compact, convex set in $\mathbb{R}^n$ for   $n=d^2-1$, Brouwer's fixed-point theorem implies that every quantum channel has at least one fixed point.  
Some sufficient conditions for $\rho_{ss}$ to be unique are described in Ref.~\cite{Wolf12}.  

We consider families of quantum channels $\E_\sigma(\cdot)$,    where $\sigma$ is a random variable with probability measure $d \sigma$ that takes on a possibly infinite set of values.  
If we introduce an additional measure $d\nu$ on $\mathcal{F}_{ss}(\E_\sigma)$ (e.g., determined by the distribution of initial states), this naturally induces a distribution of steady-state density matrices $\rho_{\sigma\nu}$ with measure $d \sigma d\nu$.  

In the language of quantum channels, random circuit models are defined as a sequence of independent random variables ${\sigma}=(\sigma_1,\ldots,\sigma_n)$  such that
\be
\M_{n}(\rho) = \E_{\sigma_n} \circ \cdots \circ \E_{\sigma_1} (\rho),
\ee
where the $\E_{\sigma_i}$ are assumed to be drawn from identical distributions of random quantum channels.  We focus on stationary random circuits defined by the property that $\lim_{n \to \infty} \M_n(X) = \rho_\sigma \trace(X)$ with a probability that converges to one, where $\rho_\sigma$ is a density matrix that is independent of the initial state $X$ for all $X \in M_d(\C)$.
Such stationary random circuits have the convenient property that, for sufficiently large $n$, $\M_n(\cdot)$ induces a unique distribution of density matrices $\rho_\nu$ with measure $d \nu$, which is stationary in the sense that $\E_\sigma(\rho_\nu)$ is also distributed with measure $d \nu$.  This property leads to the steady-state equations for the average replicated density matrices   
\be \label{eqn:rhobar}
\overline{\rho\otimes \cdots \otimes\rho} = \int d \sigma \int d\nu\, \E_\sigma(\rho_\nu) \otimes \cdots \otimes  \E_\sigma(\rho_\nu) .
\ee

For a given distribution of quantum channels $\E_\sigma$, the following two theorems are useful in determining whether the associated random circuit $\M_n(\cdot)$ is stationary.  First, we prove that a wide class of random circuit models are in fact stationary:  
\begin{theorem}  \label{thm:markov}
Let $\M_n(\cdot)$ be a random circuit that satisfies the following properties: (i) There exists $m\in \N$ such that, with finite probability, $\M_m(\cdot)$ has a one-dimensional set of fixed points, and (ii) $\M_n$ is almost surely diagonalizable for every $n$.   Then, $\M_n(\cdot)$ is a stationary random circuit.
\end{theorem}
\begin{proof}
Condition (i) implies that the quantum channel $\M_n(\cdot)$ has a one-dimensional set of fixed points with a probability that converges to 1 with increasing $n$.  For a given sequence $\sigma$ of sufficiently large length, we denote the fixed point as $\rho_\sigma$.  Condition (ii) implies that we can then almost surely represent $\M_n(\cdot)$ as 
\be \label{eqn:mn}
{ \M}_n =P_\infty^\sigma+ \sum_{k=1}^{d^2-1} \lambda_k^\sigma  P_k^\sigma  ,
\ee
where $\lambda_k^{\sigma}$ are the eigenvalues of $\M_n(\cdot)$ with magnitude strictly less than 1, $P_k^\sigma$ are projectors into the eigenspaces,  $P_\infty^\sigma X = \rho_\sigma \trace(X)$ for all $X \in M_d(\mathbb{C})$, and $P_\infty^\sigma+ \sum_k P_k^\sigma = \id$.
 Note that the second term in Eq.~(\ref{eqn:mn}) is always traceless when acting on $X\in M_d(\C)$  due to the fact that $\M_n(\cdot)$ preserves the trace.  This condition implies that 
\be
P_\infty^{\sigma'} \sum_k \lambda_k^{\sigma} P_k^{\sigma} =0,
\ee for all $\sigma=(\sigma_1,\ldots,\sigma_n)$ and $\sigma'=(\sigma_1',\ldots,\sigma_n')$.

Now consider the limit as $q \to \infty$ of the quantum channels
\be \label{eqn:mqn}
\hat{\M}_{q n} = P_\infty^{\mu_q}+  \sum_k \lambda_k^{\mu_q} P_k^{\mu_q} P_\infty^{\mu_{q-1}} + \prod_{\ell=1}^{q} \sum_k \lambda_k^{\mu_\ell} P_k^{\mu_\ell}  ,
\ee
where $\mu_\ell=(\sigma_{(\ell-1)n+1},\ldots,\sigma_{\ell n})$.    Denote by $\lambda_{\rm max}^{\mu_\ell} = \max_k(|\lambda_k^{\mu_\ell}|)$ the maximum magnitude of the eigenvalues, now
\be
\begin{split}
||\prod_{\ell=1}^{q} \sum_k \lambda_k^{\mu_\ell} P_k^{\mu_\ell} ||_\infty & \le \prod_{\ell=1}^q ||\sum_k \lambda_k^{\mu_\ell} P_k^{\mu_\ell}||_\infty = \prod_{\ell =1}^{q} \lambda_{\rm max}^{\mu_\ell}
\end{split}
\ee
where $|| A ||_\infty $ is the magnitude of the maximum eigenvalue of $A$.  The RHS is a strictly decreasing sequence of positive real numbers, which must converge to zero as $q\to \infty$.  Applying  condition (ii)  to $\M_{qn}(\cdot)$ with $\sigma=(\sigma_1,\ldots,\sigma_{qn})$ almost surely results in the identities 
\begin{align}
P_\infty^\sigma &=  P_\infty^{\mu_q}+  \sum_k \lambda_k^{\mu_q} P_k^{\mu_q} P_\infty^{\mu_{q-1}},\\
   \sum_k \lambda_k^{\sigma}& P_k^\sigma = \prod_{\ell=1}^{q} \sum_k \lambda_k^{\mu_\ell} P_k^{\mu_\ell}.
 \end{align}
 Now, since the maximal eigenvalue $\lambda_{\rm max}^\sigma$ converges to zero, this implies the convergence with probability one
 \be
 \lim_{n\to \infty} \M_n(\rho) = \rho_\sigma,
 \ee
\end{proof}
The following theorem and a weakened version, both proved in Ref.~\cite{Wolf12}, are helpful in determining whether a given $\M_n$ satisfies condition (i) of Theorem \ref{thm:markov} 
\begin{theorem} \label{lem:wolf1}
Let $\E: M_d(\mathbb{C}) \to M_d(\mathbb{C})$ be a quantum channel with Kraus
decomposition $\E(\cdot) = \sum_i K_i \cdot K_i^\dag$.  Denote by $\K_m \equiv {\rm Span} \{ \prod_{k=1}^m K_{i_k} \}$  the complex linear span of all degree-$m$ monomials of Kraus operators forming $\E_n(\cdot)$. Then, the
following are equivalent: 
(1) For all density matrices $\rho$ the limit $\lim_{k\to \infty} \E^k(\rho)$ exists; is independent of $\rho$ and is given by a positive-definite density matrix $\rho_\infty$. (2)There exists an $m\in \mathbb{N}$ such that, for all $n\ge m$, $\K_n = M_d(\mathbb{C})$.
\end{theorem}
The modified version of this theorem has the weaker condition that we are only guaranteed the existence of an $n\in \N$ such that $\K_n  =M_d(\C)$.  In this case, the limit $\lim_{k\to \infty} \E^k(\rho)$ exists; it is independent of $\rho$, but it is given by a positive-semidefinite density matrix $\rho_\infty$.  This version of the theorem applies to the restricted region of the phase diagram with $p_1=0$, where, for $\delta=1$, the NEASs are given by pure states, i.e., rank-one density matrices.\\

\section{Symmetric simple exclusion process}
\label{app:SSEP}
Here, we review some basic properties of the SSEP \cite{Derrida07}. The master equation for the configuration probability $P(\bm{\tau})$ of the SSEP takes the form (see also Eq.~(\ref{eqn:dpdt}))
\be \label{eqn:dpdt2}
\frac{d P(\bm{\tau})}{d t } = \sum_{\{\sigma_i\}} W_{\bm{\tau}}^{\bm \sigma} P(\bm{\sigma}).
\ee
The transition matrix $W_{\bm{\tau}}^{\bm \sigma}$ can be decomposed into two single-site boundary operators and $L-1$ two-site operators in the bulk,
\be
W_{\bm{\tau}}^{\bm{\sigma}} =  [R_L]_{\tau_1}^{\sigma_1} \otimes \mathbb{I} +[R_R]_{\tau_L}^{\sigma_L} \otimes \mathbb{I}+ \sum_{i=1}^{L-1} [W_i]_{\tau_{i i+1}}^{\sigma_i \sigma_{i+1}}\otimes \mathbb{I},
\ee
where, in the local basis $\{11,01,10,00\}$, $W_i$ is equal to an effective hopping matrix
 \begin{align}
 W_i&=\left(\begin{array}{c c c c}
 0 & 0 & 0 &0 \\
 0& -1/2 & 1/2 & 0 \\
 0 & 1/2 & -1/2 & 0\\
 0 & 0 & 0 &0
 \end{array} \right),
 \end{align}
while the boundary transition matrices for antisymmetric reservoirs take the form
 \be
R_{L/R}=\frac{1}{4}\left(\begin{array}{c c }
-1\pm\delta & 1\pm\delta  \\
1\mp\delta & -1\mp\delta
\end{array} \right).
\ee
 The NESS has an exact solution in terms of a matrix-product-state (MPS) representation.  Specifically,
 \be \label{eqn:pcliff}
 P(\tau_1,\ldots, \tau_L)=\bra{U} A_{\tau_1} \cdots A_{\tau_L} \ket{V},
 \ee
 where $A_{L/R}$ are infinite-dimensional matrices and $\ket{U,V}$ are vectors that satisfy the algebraic equations
 \begin{align}
 [A_L, &A_R]= 2(A_L+A_R), \\
  \ket{V}&= \Big(- \frac{1+\delta}{4} A_L + \frac{1-\delta}{4}A_R\Big) \ket{V} , \\
\bra{U}& = \bra{U} \Big( -\frac{1+\delta}{4} A_R + \frac{1-\delta}{4} A_L \Big)  .
\end{align}
From these relations, it directly follows that Eq.~(\ref{eqn:pcliff}) satisfies (\ref{eqn:dpdt})/(\ref{eqn:dpdt2}).  This algebra can be used to efficiently compute the probability amplitude of any configuration by recursively reducing the number of $A_{\tau_i}$ matrices in the representation of $P(\tau_1,\dots,\tau_L)$ from $L$ to  zero.  Explicit representations of these matrices are given in Ref.~\cite{Derrida93}. The one- and two-point connected correlations functions for $i<j$  are given by \cite{Derrida07}
\begin{align} \label{eqn:tx}
\mean{\tau_{i}} &= \frac{1+\delta}{2} - \delta \frac{i}{L+1}, \\ \label{eqn:txty}
\tau_{ij} &= - \frac{i(L+1-j)}{(L+1)^2 L} \delta^2, 
\end{align}
where $\tau_{ij} \equiv  \mean{\tau_{i}\tau_{j}} -  \mean{\tau_{i}}\mean{\tau_{j}}$.  Note that the two-point function is negative due to the repulsive hard-core interactions between particles.
  In the scaling limit, the connected three-point function also has a simple form ($x<y<z$)
\begin{align}
 \tau(x,y,z) = - 2 \frac{x (1-2y)(1-z)}{L^2} \delta^3.
\end{align}

\section{Scaling limit for phase II}
\label{app:NIFder}

In this appendix, we provide more details of the derivation of the volume-law correction to the average entropy and mutual information for the noninteracting fermion random circuit.  First, we note that the free-fermion gates acting on sites $k$ and $k+1$  obey the relation
\be
U c_i^\dag U^\dag = \sum_{j} V_{ij} c_j^\dag
\ee
where $V_{ij}$ acts as the identity on most sites except when $j$ is equal to $k$ or $k+1$, in which case it is a Haar random matrix on U(2) with probability $p_1$, permutes sites $k$ and $k+1$ with probability $(1-p_1)/2$ due to the action of the iSWAP, or is the identity with probability $(1-p_1)/2$.  The action of the reservoir can also be simply accounted for by combining the action of the unitary gate, swapping with the reservoir, and tracing over the reservoir into a Kraus operator acting on the density matrix; however, in the scaling limit ($L\to \infty$) a detailed consideration of the boundary operators is not necessary to lowest order in $1/L$ as they only serve to set boundary conditions on the course-grained correlation functions.   

In deriving time-evolution equations for  the covariance matrix,
\begin{align}
A_{ij} \equiv (1-\delta_{ij})  \overline{\lvert \langle c_i^\dag c_j \rangle \lvert^2 } + \delta_{ij} \big( \overline{\mean{n_i}^2 } - \overline{\mean{n_i}}^2\big),
\end{align}
one finds that circuit averaging will couple $A_{ij}$ to density-density fluctuations $\overline{\mean{ n_i} \mean{n_j}}$ for $i \ne j$.  However, since every NEAS is a Gaussian fermionic state, we can use Wick's theorem  to find a relation between the density-density correlations and $A_{ij}$ using the mapping of  $\bar{\rho}$ to  the SSEP
\be
\begin{split}
\overline{ \mean{n_i n_j} }&=\overline{\mean{n_i}} ~ \overline{\mean{n_j}}  + \mean{\tau_i \tau_j}_c \\
&= \overline{\langle c_i^\dag c_i c_j^\dag c_j \rangle }  =\overline{\mean{ n_i} \mean{n_j}} - \overline{\lvert \langle c_i^\dag c_j \rangle \lvert^2} , \label{eqn:4pt} 
\end{split}
\ee
where $\mean{\tau_i \tau_j}_c$ is given in Eq.~(\ref{eqn:txty}).  Using the definition of $A_{ij}$, we find that the average density-density correlations are given by
\be
\overline{\mean{ n_i} \mean{n_j}} - \overline{\mean{n_i}}~  \overline{\mean{n_j}}  =A_{ij} + \tau_{ij} .
\ee
From our solution for $A_{ij} = - \tau_{ij}$, we find that these correlations exactly vanish to lowest order in $1/L$.  

\begin{widetext}
Denoting averages over random circuits with an overbar,  the time-evolution equations for the covariance matrix elements are given by
\be
\begin{split}
\overline{\langle  c_{i_1}^\dag c_{i_2} \rangle \langle c_{i_3}^\dag c_{i_4} \rangle}(t+\delta t)  =  \sum_{k,\{m_\ell \} } \int \frac{ d \mu }{L}V_{i_1 m_1}^{k \sigma} V_{i_2 m_2}^{*k \sigma} V_{i_3 m_3}^{k \sigma} V_{i_4 m_4}^{*k \sigma}  \overline{\langle c_{m_1}^\dag c_{m_2} \rangle \langle c_{m_3}^\dag c_{m_4} \rangle}(t),
\end{split}
\ee
where $V_{ij}^{k \sigma}$ acts on sites $k$ and $k+1$ and $d \mu $ is the  probability measure for the randomly chosen set of gates on that site.   To compute this average, we can use standard formulas for Haar averages of tensor products of matrices on $U(n$) \cite{Nahum17}.  In particular, for a given $k$, when $i_\ell$ and $m_\ell$ are all equal to either $k$ or $k+1$, we have the identity
\be
\begin{split}
\int d \mu V_{i_1 m_1}^{k \sigma} V_{i_2 m_2}^{*k \sigma} V_{i_3 m_3}^{k \sigma} V_{i_4 m_4}^{*k \sigma} & =  \frac{p_1}{3} \Big[ \delta_{i_1 i_2} \delta_{i_3 i_4} \delta_{m_1 m_2} \delta_{m_3 m_4} + \delta_{i_1 i_4} \delta_{i_2 i_3} \delta_{m_1 m_4} \delta_{m_2 m_3}  \\
&- \frac{1}{2} \Big( \delta_{i_1 i_4}\delta_{i_2 i_3} \delta_{m_1 m_2} \delta_{m_3 m_4} + \delta_{i_1 i_2} \delta_{i_3 i_4} \delta_{m_1 m_4} \delta_{m_2 m_3} \Big) \Big] \\
&+ \frac{1-p_1}{2} \Big( \prod_{\ell=1}^4\delta_{i_\ell m_\ell}  +  \prod_{\ell=1}^4 F^{k k+1}_{i_\ell m_\ell} \Big), 
\end{split}
\ee
where $F^{k k+1}$ is a permutation matrix for sites $k$ and $k+1$.  On the other hand, if only a pair of $i_\ell$ indices are equal to $k$ or $k+1$ (e.g., $i_1$ and $i_3$), we instead have the identity
\be
\begin{split}
\int d \mu V_{i_1 m_1}^{k \sigma} V_{i_2 m_2}^{*k \sigma} V_{i_3 m_3}^{k \sigma} V_{i_4 m_4}^{*k \sigma}  &= \frac{1}{2} \Big( \prod_{\ell=1}^4\delta_{i_\ell m_\ell}  + F^{k k+1}_{i_1 m_1} F^{k k+1}_{i_3 m_3} \delta_{i_2 m_2} \delta_{i_4 m_4}  \Big).
\end{split}
\ee
Using these relations together with Eq.~(\ref{eqn:4pt}), we find the steady-state equations for the covariance matrix in the bulk $ 2\le i\le L-2$ and $i+2\le j\le L-1$,
\begin{align}  \label{eqn:aii}
0 &= - \Big( \frac{1}{2} + \frac{p_1}{6}\Big) A_{ii} + \Big(\frac{1}{4}- \frac{p_1}{12} \Big) (A_{i-1i-1}+A_{i+1 i+1}) + \frac{p_1}{3} (A_{i i+1} + A_{i-1 i} )  \\ \nonumber
&+ \frac{ \delta^2}{2 (L+1)^2}  - \frac{2 i (L+1) - 2i^2 -1}{6 L (L+1)^2}  \delta^2 p_1, \\  \label{eqn:daiipdt}
0 &= - \Big( \frac{1}{3} + \frac{2 p_1}{9}\Big) A_{ii+1} +  \frac{p_1}{18}  (A_{ii}+A_{i+1 i+1}) + \frac{1}{6} (A_{i-1 i+1} + A_{i i+2} )  + \frac{2 i L - 2 i^2 + L}{18 L (L+1)^2} \delta^2 p_1 \\
0 & = -\frac{1}{2}A_{ij} + \frac{1}{8} (A_{i-1 j}+A_{i+1 j} + A_{i j-1}+A_{i, j+1}).
\end{align}
\end{widetext}

Before examining the steady-state solutions  for finite $p_1$, it is constructive to examine their behavior for $p_1=0$, where we have an analytic solution for the NEAS density matrices and their probability distribution.  In this case, we see that $A_{ii}$ decouples from $A_{i i+1}$, which becomes undriven.  As a result, $A_{ij}=0$ for $j \ne i$ and we are left with a discrete diffusion equation for $A_{ii}$,
\be
A_{i-1 i-1} - 2 A_{ii} + A_{i+1 i+1}  = - \frac{2\delta^2}{(L+1)^2}.
\ee
We define $h(x)=A_{xL,xL}$, which, in the scaling limit, satisfies the boundary conditions $h(0)=h(1)=0$, leading to the solution
\be \label{eqn:Axx}
h(x) = x (1-x) \delta^2+ O(L^{-1}).
\ee
An alternative derivation of this result is to use the mapping of $\mathbb{P}(\rho_{\bm{\tau}} )$ to the SSEP with fully pseudospin-polarized reservoirs to write
\be
\overline{\mean{n_i}^2} = \Big( \frac{1+ \delta}{2} \Big)^2 \mean{\tau_i} + \Big( \frac{1- \delta}{2}\Big)^2 (1 -\mean{\tau_i}) ,
\ee
where $\mean{\tau_i} = 1 - i/(L+1)$, which agrees with (\ref{eqn:Axx}) in the scaling limit after subtracting $\overline{\mean{n_i}}^2 $.  Using this solution to compute the correction to the entropy we find
\be
 S(\rho_{\rm LE})- \overline{S(\rho_{\bm \tau})}  =   \frac{\delta^2}{3} L + O(L^0) + O(\delta^3),
  \ee
  which agrees with Eq.~(\ref{eqn:delSssep}) expanded to order $\delta^2$.
  
  For finite $p_1$, the solution to the steady state exhibits  crossover behavior in the scaling limit because the $A_{i j}$ correlations are now sourced by the current along the diagonal.  To solve for $A_{ij}$, we first  invert the steady-state solution of Eq.~(\ref{eqn:aii})  and find
    \be
\begin{split}
  A_{ii} &= \sqrt{\frac{p_1}{3}} \sum_j e^{- \lambda \abs{i - j}} (A_{j j+1} + A_{j-1 j}) \\
&   + \frac{ 3 (L+1) - 2 i (L+1 - i) p_1}{2 L (L+1)^2 p_1} \delta^2,
\end{split}
  \ee
  where $\cosh \lambda = \frac{3+p_1}{3-p_1}$.  This expression is valid in the bulk up to exponentially small corrections on the order of $e^{-\lambda L}$.
Notice that the second term is nonperturbative in $p_1$, which is consistent with the  nonanalytic behavior we find in the entropy and mutual information.  Defining $a(x,y)=A_{xL,yL+1}$, we can rewrite the first term as
\be
\begin{split}
\sqrt{\frac{p_1}{3}} &\sum_j e^{- \lambda \abs{i - j}} (A_{j j+1} + A_{j-1 j}) \\
&\approx 4 a(x,x)+\frac{3}{L^2 p_1} (\del_x+\del_y)^2 a(x,y)\lvert_{x=y}.
\end{split}
\ee
Inserting this identity into Eq.~(\ref{eqn:daiipdt}) leads to the coarse-grained diffusion equation for $a(x,y)$, 
\begin{align}
 \nabla^2 a = - \frac{2 \delta(x-y)}{L} \big[ \delta^2+  (\del_x+\del_y)^2 a \big].
 \end{align}
The second term on the rhs contributes an $O(L^{-2})$ correction to $a(x,y)$, which is why we neglected it in the discussion in the main text.

 \section{Average replicated density matrix in chaotic phase III - $abc$ Model}
 \label{app:phaseIII}

 In this appendix, we solve for the one- and two-point correlation functions of the time-averaged replicated density matrix $\overline{\rho_{\rm NEAS} \otimes \rho_{\rm NEAS}}$ perturbatively in $\delta$ and $L^{-1}$. Using the ansatz that the higher-order connected correlation functions scale with higher powers of $\delta$, these correlation functions then allow us to compute the deviations of the average entropy and mutual information away from $\bar{\rho}$.  As discussed in the main text, we find that the entropy converges to that of local equilibrium, while the mutual information obeys an area law to second order in $\delta$.
 
To describe the dynamics of the NEAS  $\rho_\nu$, we first decompose the state into a complete set of orthogonal, Hermitian operators $S$ normalized according to $\trace(S S'^\dag)= \delta_{SS'}$
\begin{align}
\rho_\nu &=  \sum_S a_S^\nu S , ~~a_S^\nu=\trace(\rho_\nu S^\dag).
\end{align}
The random channels that make up the random circuit map the state $\rho_\nu$ to a different NEAS $\E_\sigma(\rho_\nu)$.
Defining the coefficients $e_{SS'}^\sigma$ by 
\be
\E_\sigma(S)=\sum_{S'} e_{SS'}^\sigma S',
\ee
the density matrix is updated as
\begin{align}
\E_\sigma(\rho_\nu)&=\sum_{SS'}e_{SS'}^\sigma a_{S'}^\nu S.
\end{align}
Then the average moments in the NEASs satisfy
\be \label{eqn:moments}
\overline{a_{S_1}^\nu \cdots a_{S_n}^\nu}  = \sum_{S_i'} \overline{e_{S_1S_1'}^\sigma \cdots e_{S_nS_n'}^\sigma} ~ \overline{a_{S_1'}^\nu \cdots a_{S_n'}^\nu}.
\ee
As a result, the matrix $\overline{ \rho_{\rm NEAS}^{\otimes n}}$ is given by the NESS of a stochastic process in the space of operator strings with  the transition matrix $W_{SS'}=\overline{e_{S_1S_1'}^\sigma \cdots e_{S_nS_n'}^\sigma} - \delta_{SS'}$ \cite{Nahum17}.
 
  To compute the average entropy and mutual information to lowest order in $\delta$, it is sufficient to have access to only the doubled matrix $\overline{\rho_{\rm NEAS} \otimes \rho_{\rm NEAS}}$.   To solve for this operator, we first reduce to a subset of the full Hilbert space that describes the NEASs.  Unitarity of the bulk dynamics implies that $\rho^2$ evolves with the same average dynamics as $\rho$.  This implies that $\overline{\rho_{\rm NEAS}^2}$ has only  diagonal components in the local $z$ basis, placing constraints on  which second-order moments $\overline{a_S a_{S'}}$ are allowed to be nonzero.  In particular, the only allowed operator string pairs $\binom{S}{S'}$ at each site are one of the following six  pairs
\be
\left( \begin{array}{c}
S \\
S'
\end{array}\right)=
\left( \begin{array}{c}
\ldots \ell \ldots r \ldots u \ldots d \ldots u\ldots d\ldots  \\
\ldots r\ldots  \ell\ldots u\ldots d\ldots d\ldots u \ldots
\end{array}\right),
 \ee
where the single-site operators are taken as $u =(1+\sigma^z)/2,$ $d=(1-\sigma_z)/2,$ $\ell=\sigma^-$, and $r=\sigma^+$.  This simplification reduces the size of the Hilbert space needed to represent $\overline{\rho_{\rm NEAS} \otimes \rho_{\rm NEAS}}$ from $16^L$ to $6^L$.

To describe the dynamics in this basis we make a formal mapping of each of these on-site operator pairs to one of three classes of spin-1/2 particles: $a_\uparrow=\binom{\ell}{r},~a_{\downarrow}=\binom{r}{\ell},~b_\uparrow=\binom{u}{u},~b_\downarrow=\binom{d}{d}, ~c_\uparrow=\binom{u}{d},$ and $c_\downarrow=\binom{d}{u}$.  
 It is further convenient to make a Jordan-Wigner transformation on the $\ell_i$ and $r_i$ operators, which leads to an anticommutation relation for the $a$ and $c$ particles.
We can then map the solution for $\overline{\rho\otimes \rho}$ to the NESS of a six-species symmetric exclusion process in the space
$
(\mu_1^{s_1},\ldots,\mu_N^{s_N}),
$
where $\mu_i \in \{ a,b,c \} $, $s_i \in\{ \uparrow,\downarrow \}$.  We can evaluate the transition matrix in this representation following a similar approach as for the NIF random circuit, making use of standard formulas for averages over Haar random unitaries.  We find that the amplitude of each configuration $P_L$ evolves according to
\be \label{eqn:dpn}
\begin{split}
\frac{d P_L(\bm{\mu^{s}})}{dt}  = \sum_{\bm{\sigma}} W_{\bm{\mu^{s}}}^{\bm{\sigma}} P_L(\bm{\sigma}) .
\end{split}
\ee
 The two-site transition matrix in the bulk is a $6 \times 6 \times 6 \times 6$ tensor with nonzero entries,
   \begin{widetext}
 \begin{align}
 \small{
 W_{0}=\left(\begin{array}{c c c c c c}
 -\frac{1}{2}-\gamma_1 & \frac{1}{2}-\gamma_1 & \gamma_1 & \gamma_1 & \gamma_1 & \gamma_1 \\
 \frac{1}{2}-\gamma_1 & -\frac{1}{2}-\gamma_1 & \gamma_1 & \gamma_1 & \gamma_1 & \gamma_1 \\
 \gamma_1 & \gamma_1 & -\frac{1}{2}-\gamma_1 & \frac{1}{2}-\gamma_1 & -\gamma_1 & -\gamma_1  \\
  \gamma_1 & \gamma_1 & \frac{1}{2}-\gamma_1 & -\frac{1}{2}-\gamma_1 & -\gamma_1 & -\gamma_1  \\
  \gamma_1 & \gamma_1 &  -\gamma_1 & -\gamma_1 & -\frac{1}{2}-\gamma_1 & \frac{1}{2}-\gamma_1 \\
  \gamma_1 & \gamma_1 &  -\gamma_1 & -\gamma_1 & \frac{1}{2}-\gamma_1 & -\frac{1}{2}-\gamma_1 
 \end{array}\right), }
 \end{align}
 \begin{align}
 W_{a_s b_{s'}}&=W_{b_s c_{s'}}=\left(\begin{array}{c c}
 -\frac{1}{2} & \frac{1}{2} \\
\frac{1}{2} & -\frac{1}{2}
 \end{array} \right), ~W_{a_s c_{s'}}=\left(\begin{array}{c c}
-\frac{1+\gamma_2}{2}  & \frac{1-\gamma_2}{2} \\
\frac{1-\gamma_2}{2}& -\frac{1+\gamma_2}{2}
 \end{array} \right),
 \end{align}
 where $\gamma_1 \equiv p_1/6$, $\gamma_2 \equiv p_2$, $W_0$ acts in the subspace $\{ b_\uparrow b_\downarrow, b_\downarrow b_\uparrow, c_\uparrow c_\downarrow, c_\downarrow c_\uparrow, a_\uparrow a_\downarrow, a_\downarrow a_\uparrow \} $ and $W_{\mu_s \nu_{s'}}$ acts in the subspace $\{ \mu_s \nu_{s'}, \nu_{s'}\mu_s \}$.    
 The left boundary transition matrix is $R_L = R_L^{bc} \oplus R_L^a$ where the matrix elements of $R_L^{bc}$ in the basis $\{b_\uparrow,b_\downarrow,c_\uparrow,c_\downarrow \}$ are
  \begin{align}
 \footnotesize{R_L^{bc}=
\left(\begin{array}{c c c c c c}
-p_d^2(\frac{1}{2}+\gamma_1) - p_d p_u & p_u^2(\frac{1}{2}-\gamma_1) & \frac{p_u^2}{2}+\gamma_1 p_d p_u & \frac{p_u^2}{2}+\gamma_1 p_d p_u  \\
p_d^2(\frac{1}{2}-\gamma_1) & -p_u^2(\frac{1}{2}+\gamma_1) - p_d p_u & \frac{p_d^2}{2}+\gamma_1 p_d p_u & \frac{p_d^2}{2}+\gamma_1 p_d p_u \\
\frac{p_d p_u}{2}+\gamma_1 p_d^2 & \gamma_1 p_u^2+\frac{p_d p_u}{2} & -\frac{1}{2} + p_d p_u (1-2\gamma_1) &0  \\
\frac{p_d p_u}{2}+\gamma_1 p_d^2 & \gamma_1 p_u^2+\frac{p_d p_u}{2} & 0 &-\frac{1}{2} +  p_d p_u (1-2\gamma_1)   
\end{array} \right),}
\end{align}
  \end{widetext} 
where $p_u = (1+\delta)/2$ and $p_d = (1-\delta)/2$.  The matrix elements of $R_L^a$ in the basis $\{a_{\uparrow},a_\downarrow\}$ are
\be
R_L^a =  \left(\begin{array}{c c}
-1/2 & 0 \\
0 & -1/2 
\end{array} \right).
\ee
The right boundary transition matrix $R_R$ can be obtained by switching $p_d$ and $p_u$ in the expressions for $R_L$.
 
   From these expressions, we see that the interacting gates effectively induce a dissipative interaction between $a$ and $c$ particles.  This result can be understood intuitively because the interactions effectively couple the off-diagonal correlations of the fermions to density-density fluctuations, which damp out the off-diagonal terms.  Alternatively, the presence of these dissipative terms can be understood using the arguments given in the main text, whereby the nonconserved operators spread ballistically to the reservoirs where they are decohered, which leads to a dissipation term in the bulk.

There are several helpful identities to keep in mind when working with this representation for $\overline{\rho_{\rm NEAS} \otimes \rho_{\rm NEAS}}$.  First, the fact that $\trace[\rho] =1$ is reflected by the identity
\be
1=\overline{\trace[\rho_{\rm NEAS}]^2} = \sum_{\{ \mu_{i}^{s_i}\, : \,  \mu_i \in\{b,c\} \}}  P_L(\mu_1^{s_1},\ldots,\mu_{L}^{s_L}) .
\ee
This conservation law is explicitly preserved by the bulk and boundary transition matrices. 
Similarly we can express the average purity as
\be \label{eqn:purityabc}
\trace[ \overline{\rho_{\rm NEAS}^2}] = \sum_{\{ \mu_{i}^{s_i}\, : \, \mu_i \in\{a,b\} \}}P_L(\mu_1^{s_1},\ldots,\mu_{L}^{s_L}).
\ee
The average purity is preserved by the bulk transition matrices, which are associated with unitary dynamics;  however, purity is not conserved by the boundary matrices.  These boundary terms conserve the total number of $b$ and $c$ particles to preserve probability, but give rise to pure damping of the $a$ particles. The reduced density matrices for a subset of sites $A = \{i_1,\ldots,i_n\} \subset \{1,\ldots,L\}$ have the representation
\begin{align} 
\overline{\rho_{\rm NEAS}^A \otimes \rho_{\rm NEAS}^A} &= \sum_{ \{ \mu_{i_j}^{s_{i_j}}\, :\,  i_j \in A \}} P_A(\bm{\mu^{s}}) \bigotimes_{j=1}^n \hat{\mu}_{i_j}^{s_{i_j}} ,\\
P_A(\mu_{i_1}^{s_{i_1}},\ldots,&\, \mu_{i_n}^{s_{i_n}} ) = \sum_{\{ \mu_{i}^{s_i}\, : \, i \in A^c, \, \mu_i \in\{b,c\} \}}  P_L(\bm{\mu^{s}}) ,
\end{align}
where $A^c$ is the complement of $A$,  $\rho^A= \trace_{A^c} [\rho]$, and $\hat{\mu}_{i}^{s_{i}}$ denotes the single-site operator in the doubled space corresponding to the label $\mu_{i}^{s_{i}}$.

In order to map the dynamics of $P_L$ to a classical stochastic lattice gas, one requires that $P_L\ge 0$ for all configurations of the particles.  By definition, positivity of $P_L$ is guaranteed for configurations that contain only $a$ and $b$ particles and only $b$ and $c$ particles; however, any configuration that contains a mixture of $a$ and $c$ particles is allowed to have negative weight.  This fact should provide sufficient warning to the reader to avoid literally interpreting the dynamics of the $abc$ model as a classical stochastic lattice gas.  On the other hand, for the problem at hand and sufficiently small values of $\delta$, we find that $P_L$ does not have negative weight configurations in the large-$L$ limit.  Thus, our analysis provides an \textit{a posteriori} justification for the interpretation of  the $abc$ model as a classical stochastic lattice gas, but we do not assume the positivity of $P_L$ in the analysis below.

To compute quantities such as the entropy, purity, or mutual information, it is obviously sufficient to have access to arbitrarily high-order correlation functions.  In the case of the NIF random circuit, we took advantage of the fact that the NEASs are Gaussian fermionic states to represent the entropy in terms of two-point functions.  For the full interacting problem, such a simplification is no longer possible. Instead, we take advantage of the fact that the only correlations in the system are perturbatively suppressed in the current $J\sim\delta/L$.  This fact allows us to make an expansion of the NEAS density matrices around product states,
\be \label{eqn:rhoexp}
\begin{split}
\rho_{\rm NEAS} & = \bigotimes_{i =1}^{L} \rho_i + \sum_{\ell < m} \bigotimes_{i \ne \ell,m} \rho_i \otimes \delta \rho_{\ell m} \\
&+ \sum_{\ell<m<k} \bigotimes_{i \ne \ell,m,k}\rho_i \otimes \delta \rho_{\ell m k} + \ldots,
\end{split}
\ee
where the $\delta \rho_{\ell_1 \ldots \ell_n}$ are defined recursively as the deviation of the reduced density matrix on sites $\ell_1,\ldots, \ell_n$ from the cluster expansion.  This property implies that any partial trace over $\delta \rho_{\ell_1 \ldots \ell_n}$ is zero.  Due to the fact that the NEASs for $\delta=0$ are product states, each of these terms must be at least order $\delta$.  

To illustrate the validity of such an expansion for this class of current-driven problems we  consider a few examples, focusing on the case of the NIF random circuit for $p_2=0$.
For the NESS of the SSEP represented as a density matrix, one can show that these connected correlation functions satisfy the scaling \cite{Derrida07}
\be
\delta  \rho_{\ell_1 \ldots \ell_n} \sim \frac{\delta^n}{L^{n-1}}
\ee
For the discrete hopping random circuit, we saw that $\delta  \rho_{\ell_1 \ldots \ell_n} =0$ for all $n$.  In the case of the NIF random circuit, the mutual information between two halves of the chain satisfies a volume law; thus, we know these states contain significant long-range correlations.  Nevertheless, we  now argue that the expansion in Eq.~(\ref{eqn:rhoexp}), which occurs in the original spin basis instead of the fermionic basis, is still valid.  In phase II,  we found that the off-diagonal matrix elements of  $\delta \rho_{\ell m} $ scale, with high probability, as  $\delta/\sqrt{L}$, while the diagonal elements scale as $\delta^2/L^2$.   To evaluate the scaling of the higher-order connected correlation functions, we can use the representation of $\rho$ as a Gaussian fermionic state and Wick's theorem to relate these correlations to the two-point functions.  For example, the nonzero connected three-point functions for $i<j<k$ have the scaling for the off-diagonal correlations
\be
\delta \mean{r_i \ell_j u_k}  \equiv \mean{ r_i \ell_j u_k } -\mean{ r_i \ell_j } \mean{u_k}   = - G_{ik}  G_{kj} \sim \delta^2/L,
\ee
while the connected correlations on the diagonal scale as
\begin{align}
\delta \mean{ u_i u_j u_k } & =2 \textrm{Re}[G_{ki} G_{ij} G_{jk}]  \sim \delta^3/L^{3/2}.
\end{align}
More generally, we find that the off-diagonal connected correlation functions of length $n$, for $n$ even, scale as $\delta^{n/2}/L^{n/2}$, while, for $n$ odd, they scale as $\delta^{(n+1)/2}/L^{(n+1)/2}$.  For the diagonal connected correlations we instead find the scaling $\delta^n/L^{n/2}$.  Finally, we remark that an important feature of the class of NEASs we consider is that every density matrix in the ensemble has no correlations between states with different total $z$ angular momentum.  This result implies that the single-site density matrices $\rho_i$ have no off-diagonal components.
When evaluating the purity  to order $\delta^2$, we can represent 
\be \label{eqn:puritycor}
\begin{split}
\trace[&\rho_{\rm NEAS}^2]   \\
&= \trace\Big[ \Big(\bigotimes_{i =1}^{L} \rho_i + \sum_{\ell < m} \bigotimes_{i \ne \ell,m} \rho_i \otimes \delta \rho_{\ell m} \Big)^2\Big] +O(\delta^3) \\
&= \prod_i \trace[\rho_i^2] + 2 \sum_{\ell < m} \prod_{i \ne \ell,m} \trace[\rho_i^2] \trace[\rho_\ell \rho_m \delta \rho_{\ell m}] \\
&+ \sum_{\ell <m} \prod_{i  \ne \ell,m} \trace[\rho_i^2] \trace[\delta \rho_{\ell m}^2] +O (\delta^3).
\end{split}
\ee
Here, we were able to drop several terms because of the scaling,
\begin{align}
\trace[\rho_{\ell} \rho_m \delta \rho_{\ell m}] \trace[\rho_{\ell'} \rho_{m'} \delta \rho_{\ell' m'}] \sim \delta^4/L^4, \\
\trace[\rho_{\ell} \rho_{m} \rho_{n} \delta \rho_{\ell m n} ] \sim \delta^3/L^{3/2} ,\\
\trace[\rho_{\ell} \rho_{m} \rho_{n} \rho_k \delta \rho_{\ell m n k } ] \sim \delta^4/L^{2}.
\end{align}
These scalings follow because multiplication by the single-site density matrices does not map off-diagonal terms of $\delta \rho_{\ell_1\cdots\ell_n}$ onto the diagonal, while the diagonal components have the scaling $\delta \rho_{\ell m} \sim \delta^2/L^2$ and $\delta \rho_{l m k} \sim \delta^3/L^{3/2}$.  Thus, despite the volume-law mutual information in phase II, quantities such as the entropy and mutual information can be computed  with  access only to nonextensive correlation functions of $\rho_{\rm NEAS}$ and without knowledge of the Gaussian fermionic structure of the NEASs.  The error in the $n$th order truncation scales as $\delta^n$.  

For the quantum chaotic phase, it is more convenient to apply the ansatz in Eq.~(\ref{eqn:rhoexp}) to $\overline{\rho_{\rm NEAS} \otimes \rho_{\rm NEAS}}$ instead of each instance of the random circuit. 
 This is because the averaging over random circuits can induce correlations even between product-state density matrices.  For example, for $p_1=0$, the  NEASs are pure product states, while $\mathbb{P}(\rho_{\rm NEAS})$ encodes all the long-range correlations of the SSEP (see Appendix \ref{app:SSEP}).      
 
 Based on Eq.~(\ref{eqn:purityabc}) and the symmetries of the problem for $m_L = - m_R$, to compute the average purity, we need to find the correlation functions
\begin{align}
\bar{b}_i^{s}& \equiv P_{\{ i \} }(b_i^s), \\
\delta b_{ij}^{\uparrow} &\equiv P_{\{ i ,j \}} (b_i^\uparrow,b_j^\uparrow) - \bar{b}_i^\uparrow \bar{b}_j^\uparrow, \\
\delta b_{ij}^{0} &\equiv P_{\{ i ,j \}} (b_i^\uparrow,b_j^\downarrow) - \bar{b}_i^\uparrow \bar{b}_j^\downarrow, \\
a_{ij} &\equiv P_{\{ i,j \}}(a_i^\uparrow, a_j^\downarrow).
\end{align}
   These correlation functions are subject to several mathematical constraints. Conservation of probability implies that 
\begin{align}
2 \bar{c}_i + \bar{b}_i^\uparrow+\bar{b}_i^\downarrow =1 , 
\end{align}
where $\bar{c}_i  \equiv P_{\{ i \} }(c_i^s)$ is independent of $s$ by the definition of the $c$ particles.  A second constraint can be found using the fact that tracing over the second replica reduces the doubled density matrix to $\bar{\rho}$.  For the single-site density matrices, this implies 
\be \label{eqn:bpc}
\bar{b}_i^\uparrow+\bar{c}_i = \mean{\tau_i} = \frac{1+\delta}{2} - \delta \frac{i}{L+1},
\ee
where we made use of Eq.~(\ref{eqn:tx}).  
Expanding in $\delta$,  the single-site correlation functions take the form
\begin{align}
\bar{b}_i^\uparrow = \frac{1}{4} + \delta f_i + \delta^2 g_i +O(\delta^3), \\
\bar{b}_i^\downarrow = \frac{1}{4} - \delta f_i + \delta^2 g_i +O(\delta^3),
\end{align}
which follows because sending $\delta \to - \delta$  maps $\bar{b}_i^\uparrow \to \bar{b}_i^\downarrow$.  We can use Eq.~(\ref{eqn:bpc}) to constrain all the terms with odd powers of $\delta$ in this sum.  In particular, the only nonzero term with an odd power of $\delta$ is 
\be
f_i = \frac{1}{2} - \frac{i}{L+1},
\ee
which is an exact result.

We can derive similar constraints for the two-point functions.  First, they are symmetric under the mapping $\delta \to - \delta$.  This constraint implies that the only unique two-point functions are
\be
 \delta b_{ij}^{\uparrow},~\delta b_{ij}^{0},~\delta c_{ij}^{\uparrow},~\delta c_{ij}^{0},~\delta bc_{ij},~a_{ij}
 \ee
 where $\delta c_{ij}^s$ is defined analogously to $\delta b_{ij}^s$ and $\delta bc_{ij} =  P_{\{ i ,j \}} (b_i^s,c_j^{s'}) - \bar{b}_i^s c_j^{s'}$ for all $s$ and $s'$.
Tracing over one of the sites gives the identities
 \begin{align}
 \delta bc_{ij} &= -\frac{1}{2} ( \delta b_{ij}^{\uparrow} + \delta b_{ij}^{0} ), \\
  \delta b_{ij}^{\uparrow} &+ \delta b_{ij}^{0}=  \delta c_{ij}^{\uparrow} + \delta c_{ij}^{0}.
  \end{align}
  Tracing over one of the replicas similar to Eq.~(\ref{eqn:bpc}) then leads to the relations
  \begin{align}
  \delta c^\uparrow_{ij}&= \delta b_{ij}^0 + \tau_{ij}, \\
  \delta c^0_{ij} & = \delta b_{ij}^{\uparrow} - \tau_{ij}.
  \end{align}
 As a result of these constraints, to lowest order in $\delta$ the unknown correlation functions are $
g_i, $ $\delta b_{ij}^{0,\uparrow}$, and $a_{ij}.$

  It is convenient to work in scaled coordinates and define the variables 
  \begin{align}
  h(x) &\equiv \delta^2 ( g_{Lx} - f_{Lx}^2) ,\\
     a(x,y) & \equiv a_{Lx Ly+1},\\
   b(x,y) &\equiv \delta b^\uparrow_{Lx Ly+1}+\delta b_{Lx Ly+1}^0 \\ 
   B(x,y) & \equiv \delta b^\uparrow_{Lx Ly+1} - \delta b_{Lx Ly+1}^0,\\
   \tau(x,y)& \equiv \tau_{Lx Ly+1}. \\ \nonumber
   \end{align}    
    A direct calculation to lowest order in $1/L$ and second order in $\delta$ shows that these variables satisfy  independent diffusion equations away from the diagonal ($x=y$),
\begin{align} \label{eqn:hx}
h''(x)&= - 2 \delta^2+\frac{8 \gamma_1 L^2  }{1-2 \gamma_1}   \\ \nonumber
& \times [h(x) +B(x,x)  - a(x,x)-\tau(x,x)] , \\
 \nabla^2 a(x,y) &= \frac{4 \gamma_2 L^2}{2-\gamma_2} a(x,y), \\
 ~\nabla^2 b(x,y)&= \frac{8 \gamma_1 L^2}{1-2 \gamma_1} b(x,y), ~\nabla^2 B(x,y)= 0.
\end{align}
These variables are mixed by the  boundary conditions on the diagonal,
                                 \begin{widetext}
\begin{align}
 \label{eqn:dax}
 (\del_x & - \del_y) \, a(x,y)\lvert_{x=y} \, =  \frac{4 \gamma_1 \delta^2}{L(2-\gamma_2)}  + \frac{8 \gamma_1 L}{2 - \gamma_2} [ h(x) +B(x,x)- a(x,x)-\tau(x,x)], \\ 
\label{eqn:bxy}
\Big[ \frac{ 1-2 \gamma_1}{4L}& (\del_x - \del_y) b(x,y) -\frac{\gamma_1}{L} (\del_x + \del_y) [B(x,y)+a(x,y)]   \Big]_{x=y} = -\frac{ \gamma_1  }{L} h'(x) + \frac{\gamma_1 \delta^2}{L^2}(1-2x)-  \gamma_1  b(x,x)  , \\
 \label{eqn:Bxy}
\Big[ \frac{1}{4L}&(\del_x - \del_y)B(x,y) - \frac{\gamma_1}{L}(\del_x+\del_y)[B(x,y)+a(x,y)]  + \frac{\gamma_1}{L^2} \del_x \del_y B(x,y) +\frac{\gamma_1}{2 L^2}(\del_x+\del_y)^2 a(x,y)  \Big]_{x=y} \\ \nonumber
&=\gamma_1\Big[ B(x,x) +a(x,x)- \tau(x,x)  -h(x)  - \frac{ h'(x)}{L} + \frac{ \delta^2}{2L^2}(1-4x) \Big] + \frac{1-5 \gamma_1}{4 L^2} h''(x).
\end{align}
The boundary conditions at the edges of the sample are $h(0)=h(1)=0$ and $a(0,y)=a(x,1)=b(0,y)=b(x,1)=B(0,y)=B(x,1)=0$.  
\end{widetext}

For $p_1=p_2=0$, these equations have the solution 
\be
h(x)=\delta^2 x(1-x),~B(x,x)=2 \tau(x,x),
\ee
and $a(x,y)=b(x,y)=0$.
The crossover to phase II starting from $p_2=0$ occurs at $p_1 L^2 \gg 1$.  In this limit, Eqs.~(\ref{eqn:hx})--(\ref{eqn:dax}) lead to the same solution we found in Sec.~\ref{sec:phaseII} and Appendix~\ref{app:NIFder} using the fermion representation
\be
a(x,y)=-\tau(x,y),~h(x)=\frac{(1-2\gamma_1) \delta^2}{4 \gamma_1 L^2} -\tau(x,x),
\ee
where $h(x)$ deviates from this solution only in a region of order $p_1^{-1/2}$ sites near the boundaries.
 Solving for $B(x,y)$ and $b(x,y)$ using Eqs.~(\ref{eqn:bxy}) and (\ref{eqn:Bxy}) gives the solutions 
 \be
 B(x,y)=\tau(x,y),~b(x,y)=0.
 \ee  The crossover to phase III beginning in phase II, occurs for $p_2 L^2 \gg 1$, where the exponential decay of $a(x,y)$ off the diagonal implies that $a(x,x)$ is constant except over order $p_2^{-1/2}$ sites near the boundaries.  Similar to phase II, we find the solution
\be \label{eqn:bxxiii}
B(x,x)=\tau(x,x).
\ee
Solving for the other components in phase III, (i.e., $p_{1,2}L^2 \gg 1$) gives the solutions
 \begin{align}
 h(x)&= \bigg(\frac{1-2 \gamma_1}{4 \gamma_1 }+ \frac{1}{\sqrt{2 \gamma_2(2-\gamma_2)}}\bigg) \frac{\delta^2}{L^2}, \\ \label{eqn:axyiii}
 a(x,y)&= \frac{1}{\sqrt{2 \gamma_2(2-\gamma_2)}} \frac{\delta^2}{L^2} e^{- \sqrt{\frac{2 \gamma_2}{2- \gamma_2}} L (y-x)}, 
 \end{align}
 and $b(x,y)=0$.
Away from the boundaries, Eqs.~(\ref{eqn:bxxiii})--(\ref{eqn:axyiii}) provide a complete description of $\overline{\rho_{\rm NEAS} \otimes \rho_{\rm NEAS}}$  to second order in $\delta$ and lowest order in $1/L$.  

The average deviation of the entropy from local equilibrium and the mutual information between the left and right halves of the chain for $p_{1,2}L^2 \gg 1$ are then given by
 \begin{align} \nonumber
\Delta &S  = 2L \int_0^1 dx h(x) + 4L^2 \int_0^1 dy\int_0^y dx \,  a(x,y) \\ 
&=\frac{\alpha_1 \delta^2}{L p_1} + \frac{\alpha_2 \delta^2}{L p_2}, \\
I&(L:R)  = 4 L^2 \int_{1/2}^{1} dy \int_{0}^{1/2} dx\,  a(x,y)= \frac{\alpha_3 \delta^2}{L^2 p_2^{3/2} }.
\end{align}
Here, the coefficients $\alpha_i$ are all order one and take the explicit values
\begin{align}
\alpha_1&= 3-p_1,~\alpha_2=2+\sqrt{\frac{2 p_2}{2-p_2}}, \\
\alpha_3 &= \sqrt{2(2-p_2)} .
\end{align} 
Consequently, the average entropy and mutual information of the NEAS density matrices are equal to those of local equilibrium in the scaling limit.

\section{Open random quantum circuit without charge conservation}
\label{app:open}

To compare to the random circuit with conservation laws studied in the main text, here we analyze an open random circuit without conservation laws acting on a spin chain with one end of the chain coupled to a zero-entropy, spin-polarized reservoir.  
In this case, we find that the NEASs have, on average, one bit of entropy less than the infinite-temperature state.  This result can be understood intuitively because as the reservoir becomes entangled with the system, it is also injecting known pure states.  As shown in Fig.~\ref{fig:haarrand}, this reduces the entropy of the system by one bit, which is spread nonlocally across the entire chain due to the chaotic dynamics in the bulk.  

 \begin{figure}[tb]
\begin{center}
\includegraphics[width=0.45 \textwidth]{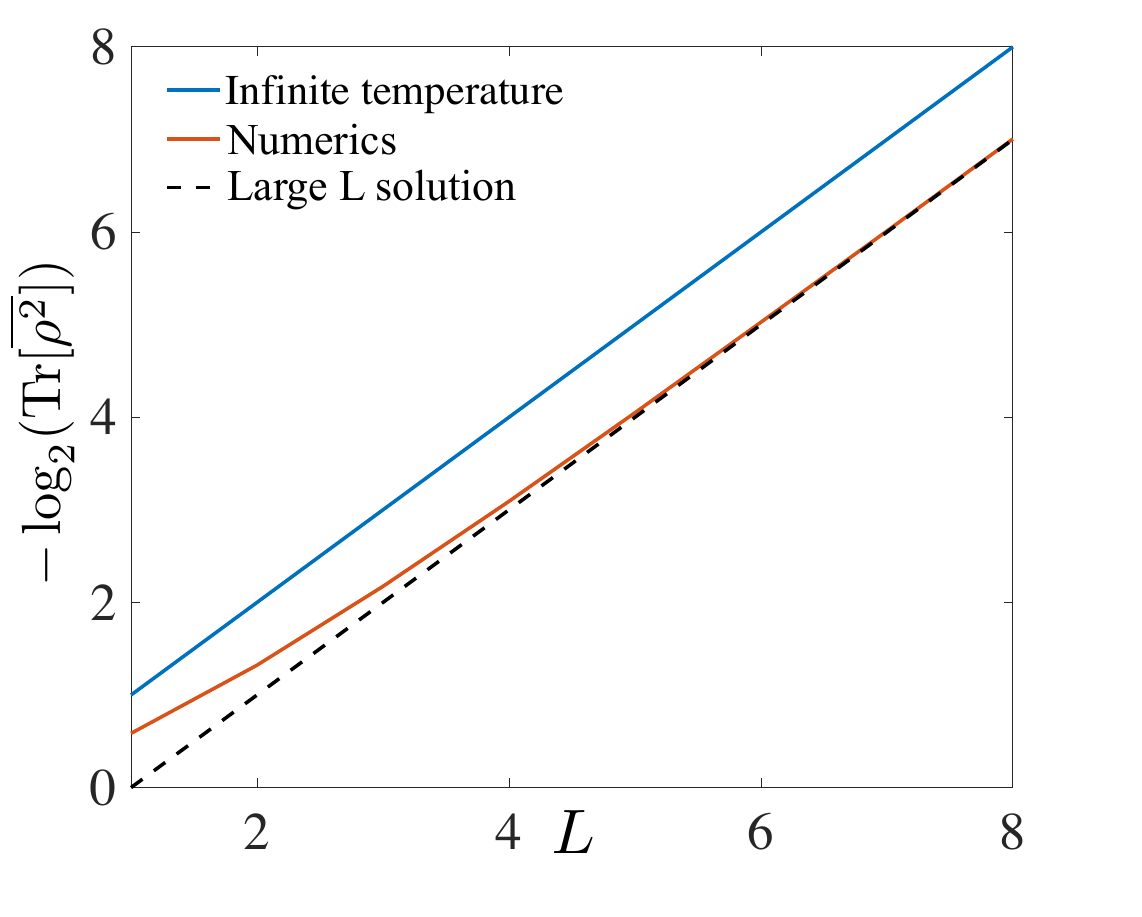}
\caption{Average purity of the NEASs for the open random circuit without charge conservation coupled to a zero-temperature reservoir. The entropy is, on average, reduced by one bit compared to the infinite-temperature state in the large-$L$ limit. }
\label{fig:haarrand}
\end{center}
\end{figure}

The random circuit is composed of Haar random two-qubit unitaries $U^\mu_{ii+1}$ acting on each pair of nearest neighbors.  After a unitary is applied on  site 0, this site is swapped with a spin from the reservoir. We can write the combined action of the unitary gate on sites 0 and 1 and the swap with the reservoir as a quantum channel acting on the reduced density matrix $\rho$ for sites $1,\ldots,L$ 
\be
\begin{split}
\E_1^\mu(\rho)&= \trace_{R,0} \Big[ U_{\rm swap}^{R0} U_{01}^\mu \, \rho \otimes \lvert{\uparrow \uparrow }\rangle \langle{\uparrow \uparrow}\lvert U_{01}^{\mu \dag} U_{\rm swap}^{R0} \Big] \\
&= \sum_r K_{1r}^\mu \rho K_{1r}^{\mu \dag},
\end{split}
\ee
where $U_{\rm swap}^{R0}$ is a two-qubit swap operator acting on the reservoir and site 0.   The two Kraus operators $K_{1r}^{\mu}$ are given by the  $2\times 2$ submatrices
\be
U_{01}^\mu = \left( \begin{array}{c  c} K_{1 \uparrow}^\mu & * \\ 
K_{1\downarrow}^\mu & * \end{array} \right),
\ee
where the basis is written as $\{ \uparrow \uparrow,\uparrow\downarrow,\downarrow\uparrow,\downarrow\downarrow \}$. 
For the other sites, there is no interaction with the reservoir and the  quantum channels associated with each unitary have the representation $\E_{i}^\mu(\rho) = U_{i-1 i}^\mu \rho  U_{i-1 i}^{\mu \dag}$, $i=2,\ldots,L$.  Similar to the case analyzed in the main text, this random circuit evolves to a set of NEASs in the long-time limit.

The time-averaged density matrix is simply given by the infinite-temperature state $\bar{\rho} =2^{-L} \id$; however, since the reservoir has zero entropy we do not expect the higher-order moments of $\rho$ to be at infinite temperature.  We can evaluate the average purity by noting that $\overline{\rho^2}$ is also proportional to $\id$, which implies that the second-order moments of the density matrix coefficients satisfy
\be
\overline{a_{S} a_{S'^\dag}} \propto \delta_{SS'}.
\ee
For this problem, we can then write the second-order moment equations in the form
\begin{align} \label{eqn:momMod1}
|a_{S}|^2 = \sum_{S'} W_{SS'} |a_{S'}|^2
\end{align}
Further simplifications are possible by noting that, due to the symmetry of the problem, the average populations of strings with $S_i=X,Y,Z$ are equal in the NEASs.  This result allows us to represent each $S\to 01001\ldots$ by binary strings, where $0$ or $1$ denotes whether a given site has a trivial or nontrivial operator \cite{Nahum17}.  In this case, the local Hilbert space for the strings is  mapped to a pseudospin 1/2, and the transition matrix takes the form
\begin{align}
W &= R \otimes \id + \sum_{i=1}^{L-1}  W_{i} \otimes \id ,~R = \left( \begin{array}{c c}
1 & 0 \\
b & 1-d 
\end{array} \right), \\
W_{i}& = \left( \begin{array}{c c c c}
1& 0 & 0 &0 \\
0 & 1/3 & 1/3 & 1/3 \\
0 & 1/3 & 1/3 & 1/3 \\
0 & 1/3 & 1/3 & 1/3
\end{array} \right),
\end{align}
where $R$ acts on site 1 in the basis $\{0, 1 \}$ and $W_{i}$ acts on sites $i$ and $i+1$ in the basis $\{00,01,10,11 \}$. The birth and death rates of the nontrivial strings on site 1 are given by $b = 0.2$ and $d=0.6$, respectively.

The operator $\overline{\rho^2}$ is fully determined by its trace  
\be
\overline{\rho^2}=  \frac{\id}{2^{2L}}\Big( 1+ \sum_S |a_S|^2\Big).
\ee
Because of the mixing induced by the boundary, the populations $p_1=\sum_S |a_{1 S}|^2$ and $p_\mu= \sum_S |a_{ \mu S}|^2$ have to be treated separately.  We can write down a reduced Markov chain describing only the populations $p_1=\sum_S |a_{1 S}|^2$, $p_\mu= \sum_S |a_{ \mu S}|^2$,  and the population in the identity string $p_0$
\be
\frac{d}{d t} \left( \begin{array}{c} 
p_0 \\
p_\mu \\
p_1 
\end{array} \right) = \left(\begin{array}{c c c}
0 & 0 & 0 \\
b & -d-\epsilon & b+\epsilon \\
0 & \epsilon & - \epsilon
\end{array} \right)  \left( \begin{array}{c} 
p_0 \\
p_\mu \\
p_1 
\end{array} \right),
\ee
This Markov chain describes the process whereby the trivial string gives birth to a nontrivial operator at site 1.   This operator then spreads throughout the system, leading to a finite mixing with $p_1$ due to evolution under $W_{1}$ at rate $\epsilon \sim 1/3$.  The strings with a nontrivial operator at site 1   die at rate $d$ and can be born from the  nontrivial operator strings with the identity at site 1 at rate $b$.   The steady state is independent of  $\epsilon$,  
\be \label{eqn:p2}
 (p_0,p_\mu,p_1) = \Big(1, \frac{b}{d-b},\frac{b}{d-b} \Big) = (1,1/2,1/2).
 \ee
The average purity is given by
 \be
-\log_2( \trace[ \overline{\rho^2}])=  (L-1) ,
 \ee
 which shows that the NEASs have, on average, one bit of entropy less than the infinite-temperature state.  We have  checked that this analysis gives excellent agreement with numerics on small chains (see Fig.~{\ref{fig:haarrand}).  There are small finite-size corrections due to the fact that there is a slightly higher probability of the information bit being lost at site one, but these decay exponentially with $L$. 

\bibliographystyle{./apsrev-nourl-title-PRX}
\bibliography{./Chaos}

\end{document}